\documentclass[11pt]{article}
\usepackage{fullpage}

\usepackage{amsmath}
\usepackage{amsfonts}
\usepackage{amsthm}
\usepackage{amssymb}
\usepackage{times}
\usepackage{braket}
\usepackage{latexsym}
\usepackage[usenames,dvipsnames]{color}
\usepackage{soul}
\usepackage{hyperref}
\usepackage{paralist}
\usepackage{bm}
\usepackage{graphicx}
\usepackage{tikz}
	\usetikzlibrary{decorations.pathreplacing}
\usepackage{stmaryrd}

\newcommand{\blnk}{ % an empty spot on the right side of the chain
\kern-1bp
\setlength{\unitlength}{10bp}
\begin{picture}(1,1)
\put(0.15,0.05){$\bigcirc$}
\end{picture}
\kern+3bp
}
\newcommand{\lmove}{ % an arrow moving left
\setlength{\unitlength}{10bp}
\kern-1bp
\begin{picture}(1,1)
\put(0.27,0.05){$\shortleftarrow$}
\put(0.15,0.05){$\bigcirc$}
\end{picture}
\kern+3bp
}
\newcommand{\turn}{ % a turn symbol, used in AGIK'09
\setlength{\unitlength}{10bp}
\kern-1bp
\begin{picture}(1,1)
\put(0.27,0.05){$\circlearrowleft$}
\put(0.15,0.05){$\bigcirc$}
\end{picture}
\kern+3bp
}
\newcommand{\insi}{ % a spot inside the qubit sequence, used in d=8
\setlength{\unitlength}{10bp}
\kern-1bp
\begin{picture}(1,1)
\put(0.4,0.05){$\circ$}
\put(0.15,0.05){$\bigcirc$}
\end{picture}
\kern+3bp
}
\newcommand{\dead}{ % a site that is done, left end of the chain
\kern-1bp
\setlength{\unitlength}{10bp}
\begin{picture}(1,1)
\put(0.27,0.05){$\times$}
\put(0.15,0.05){$\bigcirc$}
\end{picture}
\kern+3bp
}
\newcommand{\qubit}{ % a qubit holding particle, used in d=8
\setlength{\unitlength}{10bp}
\kern+1bp
\begin{picture}(1,1)
\linethickness{1bp}
\put(0,-0.2){\line(0,1){1}}
\put(0,-0.2){\line(1,0){1}}
\put(1,0.8){\line(0,-1){1}}
\put(1,0.8){\line(-1,0){1}}
\end{picture}
\kern+1bp
}
\newcommand{\gate}{ % a qubit holding particle, ACTIVE (do gate)
\setlength{\unitlength}{10bp}
\kern+1bp
\begin{picture}(1,1)
\put(0.1,0){$\blacktriangleright$}
\linethickness{1bp}
\put(0,-0.2){\line(0,1){1}}
\put(0,-0.2){\line(1,0){1}}
\put(1,0.8){\line(0,-1){1}}
\put(1,0.8){\line(-1,0){1}}
\end{picture}
\kern+1bp
}
\newcommand{\rmove}{ % a qubit holding particle, ACTIVE (move right)
\setlength{\unitlength}{10bp}
\kern+1bp
\begin{picture}(1,1)
\put(0.1,0){$\vartriangleright$}
\linethickness{1bp}
\put(0,-0.2){\line(0,1){1}}
\put(0,-0.2){\line(1,0){1}}
\put(1,0.8){\line(0,-1){1}}
\put(1,0.8){\line(-1,0){1}}
\end{picture}
\kern+1bp
}
\newcommand{\abet}{ % a beta qubit
\kern+1bp
\setlength{\unitlength}{10bp}
\begin{picture}(1,1)
\put(0.15,-0.05){$\beta$}
\linethickness{1bp}
\put(0,-0.2){\line(0,1){1}}
\put(0,-0.2){\line(1,0){1}}
\put(1,0.8){\line(0,-1){1}}
\put(1,0.8){\line(-1,0){1}}
\end{picture}
\kern+1bp
}
\newcommand{\alpa}{ % an alpha qubit
\kern+1bp
\setlength{\unitlength}{10bp}
\begin{picture}(1,1)
\put(0.15,0.05){$\alpha$}
\linethickness{1bp}
\put(0,-0.2){\line(0,1){1}}
\put(0,-0.2){\line(1,0){1}}
\put(1,0.8){\line(0,-1){1}}
\put(1,0.8){\line(-1,0){1}}
\end{picture}
\kern+1bp
}
\newcommand{\lqubit}{ % an L qubit
\kern+1bp
\setlength{\unitlength}{10bp}
\begin{picture}(1,1)
\put(0.15,0.05){$\scriptstyle{L}$}
\linethickness{1bp}
\put(0,-0.2){\line(0,1){1}}
\put(0,-0.2){\line(1,0){1}}
\put(1,0.8){\line(0,-1){1}}
\put(1,0.8){\line(-1,0){1}}
\end{picture}
\kern+1bp
}
\newcommand{\rqubit}{ % an R qubit
\kern+1bp
\setlength{\unitlength}{10bp}
\begin{picture}(1,1)
\put(0.15,0.05){$\scriptstyle{R}$}
\linethickness{1bp}
\put(0,-0.2){\line(0,1){1}}
\put(0,-0.2){\line(1,0){1}}
\put(1,0.8){\line(0,-1){1}}
\put(1,0.8){\line(-1,0){1}}
\end{picture}
\kern+1bp
}
\newcommand{\parity}{ % a vertical line
\setlength{\unitlength}{10bp}
\begin{picture}(0.3,1)
\put(0.15,-0.3){\line(0,1){1.2}}
\end{picture}
}
\newcommand{\bdry}{ % a vertical line
\setlength{\unitlength}{10bp}
\begin{picture}(0.3,1)
\put(0.07,-0.3){\line(0,1){1.2}}
\put(0.23,-0.3){\line(0,1){1.2}}
\end{picture}
}

\newcommand{\Null}{\operatorname{Null}}

\newtheorem{theorem}{Theorem}[section]

\newtheorem{lemma}[theorem]{Lemma}

\newtheorem{corollary}[theorem]{Corollary}
\newtheorem{definition}[theorem]{Definition}

\makeatletter
\newtheorem*{rep@theorem}{\rep@title}
\newcommand{\newreptheorem}[2]{%
	\newenvironment{rep#1}[1]{%
		\def\rep@title{#2 \ref{##1}}%
		\begin{rep@theorem}}%
		{\end{rep@theorem}}}
\makeatother

\newreptheorem{theorem}{Theorem}
\newreptheorem{lemma}{Lemma}

\newcommand{\comment}[1]{}

\newcommand{\tinyspace}{\mspace{1mu}}

\newcommand{\abs}[1]{\left\lvert\tinyspace #1 \tinyspace\right\rvert}

\newcommand{\norm}[1]{\left\lVert\tinyspace #1 \tinyspace\right\rVert}

\newcommand{\tr}{\operatorname{Tr}}
\newcommand{\trace}{\tr}

\newcommand{\class}[1]{\textup{#1}}
\newcommand{\prob}[1]{\textsc{#1}}

\newcommand{\Span}{\mathrm{Span}}
\newcommand{\ketbra}[2]{\ket{#1}\!\bra{#2}}
\newcommand{\brakett}[2]{\mbox{$\langle #1  | #2 \rangle$}}
\newcommand{\proj}[1]{| #1 \rangle \langle #1 |}
\DeclareMathOperator{\diag}{diag}

\newcommand{\trnorm}[1]{\norm{#1}_{\mathrm {tr}}}  % trace norm
\newcommand{\snorm}[1]{\norm{#1}_{\mathrm {\infty}}}    % spectral norm

\def\({\left(}
\def\){\right)}
\def\X{\mathcal{X}}
\def\Y{\mathcal{Y}}
\def\W{\mathcal{W}}

\newcommand{\complex}{{\mathbb C}}
\newcommand{\reals}{{\mathbb R}}

\newcommand{\R}{\mathbb{R}}

\newcommand{\cH}{\mathcal{H}}
\newcommand{\Vt}{\widetilde{V}}
\newcommand{\cE}{\mathcal{E}}
\newcommand{\cF}{\mathcal{F}}
\newcommand{\cS}{\mathcal{S}}

\newcommand{\spa}[1]{\mathcal{#1}}

\newcommand{\poly}{\textup{poly}}

\newcommand{\klh}[1][k]{{ \ensuremath{#1}\prob{-LH} }} % Optional argument for k. So \kLH[2] is 2-LH, \kLH is k-LH.
\newcommand{\psihist}{\psi_{\rm hist}}
\newcommand{\hin}{H_{\rm in}}
\newcommand{\hprop}{H_{\rm prop}}
\newcommand{\hout}{H_{\rm out}}
\newcommand{\hstab}{H_{\rm stab}}

\newcommand{\QMA}{\class{QMA}}
\newcommand{\coQMA}{\class{co-QMA}}

\newcommand{\StoqMA}{\class{StoqMA}}

\newcommand{\Poly}{\class{P}}
\newcommand{\NP}{\class{NP}}
\newcommand{\PQMA}{\class{P}^{\class{QMA}[\class{log}]}}

\newcommand{\PNPlog}{\class{P}^{\class{NP}[\class{log}]}}

\newcommand{\PQMApar}{\class{P}^{||\class{QMA}}}
\newcommand{\PNPpar}{\class{P}^{||\class{NP}}}
\newcommand{\BPP}{\class{BPP}}

\newcommand{\Plog}[1]{\class{P}^{\class{#1}[\class{log}]}}
\newcommand{\Ppar}[1]{\class{P}^{||\class{#1}}}

\newcommand{\app}{\prob{APX-SIM}}

\newcommand{\apptwo}{\prob{$\forall$-APX-SIM}}

\newcommand{\hpen}{H_{\rm pen}}
\newcommand{\din}{\Delta_{\rm in}}
\newcommand{\dprop}{\Delta_{\rm prop}}
\newcommand{\dpen}{\Delta_{\rm pen}}

\newcommand{\houti}[1][i]{H_{{\rm out},{#1}}}

\newcommand{\Fapptwo}[1][F]{\mathcal{#1}\prob{-$\forall$-APXSIM}} % Optional argument. Default is F.

\begin{document}
	
%=========================================================================================
%  BEGIN HEADER
%=========================================================================================	
\title{Oracle complexity classes and local measurements on physical Hamiltonians}

\author{
	Sevag Gharibian\footnote{University of Paderborn, Germany. Email: \href{mailto:sevag.gharibian@upb.de}{sevag.gharibian@upb.de}.}
	\and Stephen Piddock\footnote{University of Bristol, UK. Email: \href{mailto:stephen.piddock@bristol.ac.uk}{stephen.piddock@bristol.ac.uk}.}
	\and Justin Yirka\footnote{The University of Texas at Austin, USA. Email: \href{mailto:yirka@utexas.edu}{yirka@utexas.edu}.}
}

\date{September 12, 2019}

\maketitle

\begin{abstract}
	
The canonical problem for the class Quantum Merlin-Arthur (QMA) is that of estimating ground state energies of local Hamiltonians. Perhaps surprisingly, [Ambainis, CCC 2014] showed that the related, but arguably more natural, problem of simulating local measurements on ground states of local Hamiltonians ($\app$) is likely harder than QMA. Indeed, [Ambainis, CCC 2014] showed that $\app$ is $\PQMA$-complete, for $\PQMA$ the class of languages decidable by a $\class{P}$ machine making a logarithmic number of {adaptive} queries to a $\QMA$ oracle. In this work, we show that $\app$ is $\PQMA$-complete even when restricted to more physical Hamiltonians, obtaining as intermediate steps a variety of related complexity-theoretic results.

Specifically, we first give a sequence of results which together yield $\PQMA$-hardness for $\app$ on well-motivated Hamiltonians such as the 2D Heisenberg model:
\begin{itemize}
    \item We show that for $\NP, \StoqMA$, and $\QMA$ oracles, a logarithmic number of adaptive queries is equivalent to polynomially many parallel queries.
    Formally, $\Plog{NP}=\Ppar{NP}$, $\Plog{StoqMA}=\Ppar{StoqMA}$, and $\Plog{QMA}=\Ppar{QMA}$. (The result for $\NP$ was previously shown using a different proof technique.) These equalities simplify the proofs of our subsequent results.
    \item Next, we show that the hardness of $\app$ is preserved under Hamiltonian simulations (\`{a} la [Cubitt, Montanaro, Piddock, 2017]) by studying a seemingly weaker problem, $\apptwo$. As a byproduct, we obtain a full complexity classification of $\app$, showing it is complete for $\Poly, \PNPpar,\Ppar{StoqMA},$ or $\PQMApar$ depending on the Hamiltonians employed.
    \item Leveraging the above, we show that $\app$ is $\PQMA$-complete for any family of Hamiltonians which can efficiently simulate spatially sparse Hamiltonians. This implies $\app$ is $\PQMA$-complete even on physically motivated models such as the 2D Heisenberg model.
\end{itemize}

Our second focus considers 1D systems: We show that $\app$ remains $\PQMA$-complete even for local Hamiltonians on a 1D line of 8-dimensional qudits. This uses a number of ideas from above, along with replacing the ``query Hamiltonian'' of [Ambainis, CCC 2014] with a new ``sifter'' construction.

\end{abstract}

%=========================================================================================
%  SECTION: Introduction
%=========================================================================================
\section{Introduction}

The study of the low-energy states of quantum many-body systems is of fundamental physical interest. Of central focus has been the problem of estimating the ground state energy of a $k$-local Hamiltonian, known as the Local Hamiltonian problem ($\klh$). Shown by Kitaev \cite{KSV02} to be complete for the class Quantum Merlin Arthur (\QMA) (a quantum analogue of $\NP$), $\klh$ has played the role of the canonical QMA-complete problem, just as $k$-SAT is the canonical NP-complete problem. This, in turn, has given rise to the field of Quantum Hamiltonian Complexity (QHC) (see, e.g., \cite{O12,Boo14,GHLS15}), which has since explored the complexity theoretic characterization of computing properties of ground spaces \emph{beyond} estimating ground state energies. Examples have included computing ground state degeneracies~\cite{BFS11,SZ}, minimizing interaction terms yielding frustrated ground spaces~\cite{GK12}, detecting ``energy barriers'' in ground spaces~\cite{GS18,GMV17}, deciding if tensor networks represent physical quantum states~\cite{GLSW15,SMGGSPI18}, Hamiltonian sparsification~\cite{AZ18}, estimating spectral gaps of local Hamiltonians~\cite{A14,CPW15,GY16}, estimating the free energy of 1D systems~\cite{K17}, and the study of ``universal'' Hamiltonian models which can replicate the physics of any other quantum many-body system~\cite{BH17,CMP18}.

\paragraph{Approximate Simulation.} Despite the role of $k$-LH as a ``posterchild'' for Quantum Hamiltonian Complexity, in 2014 Ambainis formalized the arguably even more natural physical problem of simulating local measurements on low-energy states of a local Hamiltonian, denoting it Approximate Simulation ($\app$).
\begin{definition}[$\app(H,A,k,\ell,a,b,\delta)$~\cite{A14}]\label{dfn:apx}
	Given a $k$-local Hamiltonian $H$, an $\ell$-local observable $A$, and real numbers $a$, $b$, and $\delta$ such that $b-a\geq n^{-c}$ and $\delta\geq n^{-c'}$, for $n$ the number of qubits $H$ acts on and $c,c'>0$ some constants, decide:
	\begin{itemize}
		\item If $H$ has a ground state $\ket{\psi}$ satisfying $\bra{\psi}A\ket{\psi}\leq a$, output YES.
		\item If for all $\ket{\psi}$ satisfying $\bra{\psi}H\ket{\psi}\leq \lambda(H)+\delta$, it holds that $\bra{\psi}A\ket{\psi}\geq b$, output NO.
	\end{itemize}
\end{definition}
\noindent \emph{Motivation.} The motivation for $\app$ is as follows: Given a naturally occurring many-body quantum system with time evolution Hamiltonian $H$ (which is typically $k$-local for $k\in O(1)$), we would like to learn something about the quantum state $\ket{\psi}$ the system settles into when cooled to near absolute zero. This setting is where phenomena such as superconductivity and superfluidity manifest themselves; learning something about $\ket{\psi}$ hence potentially allows one to predict and harness such phenomena for, say, materials design. The most ``basic'' experimental approach to learning something about $\ket{\psi}$ is to attempt to prepare a physical copy of $\ket{\psi}$, and then apply a local measurement to extract information from $\ket{\psi}$.
However, given that preparing the ground state $\ket{\psi}$ of an arbitrary Hamiltonian is hard --- it would allow one to solve $\QMA$ problems --- we must wonder whether there is an easier approach. Formally, how hard is $\app$?

Perhaps surprisingly, it turns out that simulating a measurement on the ground state $\ket{\psi}$ is strictly \emph{harder than $\QMA$}. To show this, \cite{A14} proved that $\app$ is $\PQMA$-complete, for $\PQMA$ the class of languages decidable by a $\class{P}$ machine making a logarithmic number of {adaptive} queries to a $\QMA$ oracle. Why $\PQMA$ instead of $\QMA$? Intuitively, this is because $\app$ does \emph{not} include thresholds for the ground state energy as part of the input (in contrast to $\klh$). This specification of $\app$ is well-motivated; typically one does not have an estimate of the ground state energy of $H$, since such an estimate is QMA-hard to compute to begin with. (Note that if the ground state energy thresholds were included in the definition of $\app$, then the complexity of $\app$ would drop to $\QMA$.)\\

\noindent \emph{Brief background on $\PQMA$.} The class $\PQMA$ is likely strictly harder than $\QMA$, since both $\QMA$ and\footnote{To put $\coQMA$ in $\PQMA$, simply use the QMA oracle once and flip its answer using the P machine.} $\coQMA$ are contained in $\PQMA$. Thus, $\QMA\neq\PQMA$ unless $\coQMA\subseteq\QMA$ (which appears unlikely).
Just how much more difficult than QMA is $\PQMA$? Intuitively, the answer is ``slightly more difficult''. Formally, $\QMA\subseteq\PQMA\subseteq\class{PP}$~\cite{GY16} (where $\QMA\subseteq \class{A}_0\class{PP}\subseteq\class{PP}$ was known~\cite{KW00,Vy03,MW05} prior to~\cite{GY16}; note the latter containment is strict unless the Polynomial-Time Hierarchy collapses~\cite{Vy03}).

From a computer science perspective, there is an interesting relationship between $\app$ and classical constraint satisfaction problems (CSPs). The $\QMA$-complete problem $\klh$ is a quantum analogue of the $\NP$-complete problem MAX-$k$-SAT, in that the energy of a state is minimized by simultaneously satisfying as many of the $k$-local terms as possible. Classically, one might be asked whether the solution to a MAX-$k$-SAT instance satisfies some easily verifiable property, such as whether the solution has even Hamming weight; such a problem is $\PNPlog$-complete (see, e.g.,~\cite{Wag88} for a survey). $\app$ is a quantum analogue to these problems, in which we ask whether an optimal solution (the ground state) satisfies some property (expectation bounds for a specified measurement), and $\app$ is analogously $\PQMA$-complete.\\

\noindent \emph{High level direction in this work.} That $\app$ is such a natural problem arguably demands that we study its hardness given natural settings. In this regard, the original $\PQMA$-completeness result~\cite{A14} was for simulating $O(\log n)$-local observables and $O(\log n)$-local Hamiltonians, where $n$ is the number of qubits the Hamiltonian acts on. From a physical perspective, one wishes to reduce the necessary complexity, such as to $O(1)$-local observables and Hamiltonians. Hardness under this restriction was subsequently achieved~\cite{GY16}, for single-qubit observables and $5$-local Hamiltonians, by combining the ``query Hamiltonian'' construction of Ambainis~\cite{A14} with the circuit-to-Hamiltonian construction of Kitaev~\cite{KSV02}. Even arbitrary $O(1)$-local Hamiltonians, however, may be considered rather artificial in contrast to naturally occurring systems. Ideally, one wishes to make  statements along the lines of ``simulating measurements on a physical model such as the quantum Heisenberg model on a 2D lattice is harder than QMA'', or ``simulating measurements on a 1D local Hamiltonian is harder than QMA''. This is what we achieve in the current paper. Interestingly, to attain this goal, we first take a complexity theoretic turn into the world of parallel versus adaptive oracle queries.

\subsection{Parallel versus adaptive queries}\label{ssec:parVsAdap}

A natural question for oracle complexity classes is how the power of the class changes as access to the oracle is varied. In the early 1990's, it was shown~\cite{BH91,H89,Bei91} that a polynomial number of \emph{parallel} or \emph{non-adaptive} queries to an NP oracle are equivalent in power to a logarithmic number of \emph{adaptive} queries. Formally, letting $\PNPpar$ be the class of languages decidable by a $\Poly$ machine with access to polynomially many parallel queries to an $\NP$ oracle, it holds that $\PNPpar=\PNPlog$.

The direction $\Plog{C}\subseteq\Ppar{C}$ was in fact shown by \cite{Bei91} for all classes $\class{C}$. Briefly, a $\class{P}$ machine making a logarithmic number of adaptive queries to a $\class{C}$ oracle has the potential to make only polynomially many different queries, each of which can be computed beforehand in polynomial time by simulating the machine's action given each possible sequence of query answers. The values for all such queries can simply be queried in parallel by the $\Poly^{||\class{C}}$ machine.
To show the reverse direction, that $\PNPpar\subseteq\PNPlog$, one first performs binary search to determine the total number of YES queries. Then, ask whether
there exists at least that number of (provably) YES queries such that setting the corresponding query answers to YES causes the original $\Poly$ machine to accept.

We begin by considering an analogue of this question for $\PQMA$ versus $\PQMApar$ (defined as $\PNPpar$ but with a QMA oracle). The direction $\PQMA\subseteq \PQMApar$ proceeds as described above, but, in contrast, the classical technique for showing the reverse direction does not appear to carry over to the quantum setting, specifically to the setting of promise problems. As explored in~\cite{GY16}, oracles corresponding to classes of promise problems like $\QMA$ may receive queries which violate their promise (such as an instance of $\klh$ with the ground state energy within the promise gap). By definition~\cite{G06}, in such cases the oracle can respond arbitrarily, even changing its answer given repeated queries. Because of the possibility of invalid queries by the $\PQMApar$ machine, the technique of binary search fails.
To show $\PQMApar\subseteq \PQMA$, we take a different approach by instead showing a \emph{hardness} result. Specifically, we use a modification of the $\PQMA$-hardness construction of \cite{A14}, for which we require the locality improvements by~\cite{GY16}, to show that $\app$ is $\PQMApar$-hard. Combining with the known fact that $\app\in\PQMA$~\cite{A14} then yields the desired containment.

This approach includes two benefits:
\begin{itemize}
    \item First, the use of parallel, rather than adaptive, queries simplifies the ``query Hamiltonian'' construction of~\cite{A14} significantly, which we later exploit to prove hardness results about physical Hamiltonians (Theorem~\ref{thm:simofsparsehard}) and 1D Hamiltonians (Theorem~\ref{thm:apxsim-1dline-complete}).
This can also give a simpler proof of Ambainis's original claim that $\app$ is $\PQMA$-hard. Indeed, we generalize this idea to give the statement:
\begin{theorem}
	\label{thm:APXcompleteness}
	Let $\class{C}$ be a class of languages or promise problems. Let $\cF$ be a family of Hamiltonians for which \klh is $\class{C}$-complete under poly-time many-one reductions for all $k \ge 2$. Suppose $\cF$ is closed under positive linear combination of Hamiltonians, and that if $\{H_i\}_{i=1}^m\subset \cF$, then $H_{\text{cl}}+\sum_{i=1}^m \proj{1}_i\otimes H_i \in \cF$, where $H_{\text{cl}}$ is any classical Hamiltonian (i.e. diagonal in the standard basis). Then,
\[
    \Plog{C}=\Ppar{C},
\]
and $\app$ is $\Plog{C}$-complete when restricted to $k$-local Hamiltonians and observables from $\cF$.
\end{theorem}

\noindent
(The reason for the form of the expression $H_{\text{cl}}+\sum_{i=1}^m \proj{1}_i\otimes H_i$ in Theorem~\ref{thm:APXcompleteness} will become clear as we introduce the Hamiltonian constructions we use. In short, the expression suffices to encode our construction while still belonging to several interesting families $\cF$.)
Applying that $\klh$ is $\NP$-complete, $\StoqMA$-complete, and $\QMA$-complete when restricted to the families of classical, stoquastic, and arbitrary $k$-local Hamiltonians, respectively~\cite{CM16}, Theorem~\ref{thm:APXcompleteness} yields:
\begin{corollary}
	\label{cor:complexityequalities}
	$\Plog{NP}=\Ppar{NP},
	\Plog{StoqMA}=\Ppar{StoqMA}$,
	and $\Plog{QMA}=\Ppar{QMA}$.
\end{corollary}

\item Second, we base our reduction on the Cook-Levin theorem~\cite{C72,L73}, as opposed to Kitaev's circuit-to-Hamiltonian construction~\cite{KSV02} as in~\cite{GY16}. This allows us to obtain a \emph{constant} promise gap for the observable
\footnote{The constant gap is only for the input thresholds $a,b$ for the expectation value of the observable $A$. The required ``low-energy gap'' defined by the parameter $\delta$ continues to potentially scale as inverse polynomial, i.e. $\delta \geq 1/\poly$, and we note that the spectral gap of the Hamiltonian $H$ may be arbitrarily small in our constructions unless otherwise noted. Because the improved gap corresponds only to the observable, it is unclear how to apply this result to resolve questions concerning Hamiltonians with improved promise gaps, e.g. the Quantum PCP Conjecture. (As a general note, it is worth stressing here that the Quantum PCP conjecture deals with constant promise gaps, not constant spectral gaps of the Hamiltonian.)}
$A$'s threshold values (i.e. $b-a\geq \Omega(1)$, as opposed to $b-a\geq1/\poly$), even when $\|A\|=O(1)$. Further, because the core of this construction is already spatially sparse, it additionally eases proving hardness results about physical Hamiltonians (Theorem~\ref{thm:simofsparsehard}).
\end{itemize}

\subsection{The complexity of \app\ for physically motivated Hamiltonians}\label{sscn:results2}

With the simplifications that moving to parallel queries affords us (i.e. working with $\Ppar{QMA}$ versus $\Plog{QMA}$), we proceed to study $\Ppar{QMA}$-hardness for physically motivated Hamiltonians. This requires a shift of focus to \emph{simulations}, in the sense of~\cite{CMP18}, i.e. analog Hamiltonian simulations. 

Recall that Kitaev originally proved $\QMA$-hardness of $\klh$ for 5-local Hamiltonians \cite{KSV02}; this was brought down to $2$-local Hamiltonians via perturbation theory techniques~\cite{KR03,KKR06}. Since then, there has been a large body of work (e.g.~\cite{OT08,BDL11,CM16,BH17,PM17,PM18}) showing complexity theoretic hardness results for ever simpler systems, much of which uses perturbative gadgets to construct Hamiltonians which have approximately the same ground state energy as a Hamiltonian of an apparently more complicated form. Here, we wish to enable a similarly large number of results for the problem $\app$ by using the same perturbative gadget constructions and ideas of analogue simulation.

In \cite{CMP18}, the authors define a strong notion of simulation which approximately preserves almost all the important properties of a Hamiltonian, including the properties important for the problem $\klh$, and they observe that the perturbative gadget constructions used in the $\klh$ problem literature are examples of this definition of simulation. They go on to show that there exist simple families of Hamiltonians (such as the 2-qubit Heisenberg interaction) which are \emph{universal} Hamiltonians, in the sense that they can simulate all $O(1)$-local Hamiltonians efficiently.

\paragraph{How do simulations affect the complexity of \app?} Ideally, we would like to show that efficient simulations lead to reductions between classes of Hamiltonians for the problem $\app$. However, this is apparently difficult, as the definition of $\app$ is not robust to small perturbations in the eigenvalues of the system.
We instead consider a closely related, seemingly easier problem which we call $\apptwo$.
\begin{definition}[$\apptwo (H,A,k,\ell,a,b,\delta)$]	
	\label{dfn:app2}
	Given a $k$-local Hamiltonian $H$, an $\ell$-local observable $A$, and real numbers $a,b$, and $\delta$ such that satisfy $b-a \geq n^{-c}$ and $\delta \geq n^{-c'}$, for $n$ the number of qubits $H$ acts on and $c,c'>0$ some constants, decide:
	\begin{itemize}
		\item If for all $\ket{\psi}$ satisfying $\bra{\psi}H\ket{\psi}\le \lambda(H)+\delta$, it holds that $\bra{\psi}A\ket{\psi} \le a$, then output YES.
		\item If for all $\ket{\psi}$ satisfying $\bra{\psi}H\ket{\psi}\le \lambda(H)+\delta$, it holds that $\bra{\psi}A\ket{\psi} \ge b$, then output NO.
	\end{itemize}
\end{definition}
\noindent Above, we have a stronger promise in the YES case than in $\app$: namely, \emph{all} low-energy states $\ket{\psi}$ are promised to satisfy $\bra{\psi}A\ket{\psi} \le a$, as opposed to just a single ground state. Thus, $\apptwo$ is easier than $\app$, in that $\apptwo$ reduces to $\app$.
(The reduction is trivial, in that a valid instance of $\apptwo$ is already a valid instance of $\app$, with no need for modification.)
We conclude that $\apptwo$ is contained in $\PQMA$.
Furthermore, the proof of Theorem~\ref{thm:APXcompleteness} is actually sufficient to show that $\apptwo$ is $\Ppar{C}$-complete (when restricted to the corresponding family of Hamiltonians for arbitrary class $\class{C}$).

Our second result, Lemma~\ref{lem:simulations=reductions} in Section~\ref{sec:simulations}, is to prove that efficient simulations correspond to reductions between instances of $\apptwo$.
As a byproduct, we combine this result with Theorem~\ref{thm:APXcompleteness} and the universality classifications from \cite{CMP18} (cf. Corollary~\ref{cor:complexityequalities}) in order to obtain complexity classifications for the original $\app$ problem restricted to several families of Hamiltonians:

\begin{theorem}
	\label{thm:S-APXSIM}
	Let $\mathcal{S}$ be an arbitrary fixed subset of Hermitian matrices on at most 2 qubits. Then the $\app$ problem, restricted to Hamiltonians $H$ and measurements $A$ given as a linear combination of terms from $\cS$, is
	\begin{enumerate}
		\item $\class{P}$-complete, if every matrix in $\mathcal{S}$ is 1-local;
		\item $\Plog{NP}$-complete, if $\cS$ does not satisfy the previous condition and there exists $U \in SU(2)$ such that $U$ diagonalizes all 1-qubit matrices in $\mathcal{S}$ and $U^{\otimes 2}$ diagonalizes all 2-qubit matrices in $\mathcal{S}$;
		\item $\Plog{StoqMA}$-complete, if $\cS$ does not satisfy the previous condition and there exists $U \in SU(2)$ such that, for each 2-qubit matrix $H_i \in \mathcal{S}$, $U^{\otimes 2} H_i (U^{\dag})^{\otimes 2} = \alpha_i Z^{\otimes 2} + A_iI + IB_i$, where $\alpha_i \in \R$ and $A_i$, $B_i$ are arbitrary single-qubit Hermitian matrices;
		\item $\Plog{QMA}$-complete, otherwise.
	\end{enumerate}
\end{theorem}

\paragraph{Hardness of simulating local measurements on lattices and spatially sparse systems.} With the previous two main results in hand, we are in a position to show that $\apptwo$ is $\PQMA$-hard even when the Hamiltonian is restricted to a spatially sparse interaction graph (in the sense of \cite{OT08}). This is analogous to the equivalent result for $\klh$ shown in \cite{OT08}, which was crucial in showing that the Local Hamiltonian problem is $\QMA$-complete for Hamiltonians on a 2D square lattice. Formally, by exploiting the previously discussed results about parallel queries (Theorem~\ref{thm:APXcompleteness}) and simulations (Lemma~\ref{lem:simulations=reductions}) and by developing a variant of the hardness construction from Theorem~\ref{thm:APXcompleteness}, we are able to show the following:

\begin{theorem}
	\label{thm:simofsparsehard}
	Let $\mathcal{F}$ be a family of Hamiltonians which can efficiently simulate any spatially sparse Hamiltonian. Then, $\app$ is $\PQMA$-complete even when restricted to a single-qubit observable and a Hamiltonian from the family $\mathcal{F}$.
\end{theorem}

Via Theorem~\ref{thm:simofsparsehard}, we now obtain many corollaries via the long line of research using perturbative gadgets to prove $\QMA$-completeness of restricted Hamiltonians; for brevity, here we list a select few such corollaries. We note that the locality of the observable input to $\app$ may increase after simulation, but only by a constant factor which can be easily calculated based on the simulation used. For example, using the perturbative gadgets constructed in \cite{PM17}, the following is an immediate corollary of Theorem~\ref{thm:simofsparsehard}:
\begin{corollary}\label{cor:1}
The problem $\app$ is $\PQMA$-complete even when the observable $A$ is 4-local and the Hamiltonian $H$ is restricted to be of the form:
\[
	H=\sum_{(j,k) \in E} a_{(j,k)} h_{(j,k)} , \qquad \text{ where } h_{(j,k)}= \alpha X_j X_k + \beta Y_j Y_k + \gamma Z_j Z_k ,
\]
$E$ is the set of edges of a 2D square lattice, $a_{(j,k)} \in \R$, and at least two of $\alpha, \beta, \gamma$ are non-zero.
The case $\alpha=\beta=\gamma$ corresponds to $XX+YY+ZZ$, which is known as the Heisenberg interaction.
\end{corollary}

\noindent But, there is not always a blow-up in the locality of $A$, as is shown by this corollary which follows from Theorem~\ref{thm:simofsparsehard} and \cite{SV09}:
\begin{corollary}\label{cor:2}
The problem $\app$ is $\PQMA$-complete even when the observable $A$ is 1-local and the Hamiltonian $H$ is restricted to be of the form:
\[
	H=\sum_{(j,k) \in E} h_{(j,k)}+\sum_j B_j , \qquad \text{ where } h_{(j,k)}=  X_j X_k + Y_j Y_k +  Z_j Z_k ,
\]
$E$ is the set of edges of a 2D square lattice, and $B_j$ is a single qubit operator (that may depend on $j$).
\end{corollary}

\noindent Finally, we remark that recent work on the simulation power of families of qudit Hamiltonians \cite{PM18} can be used to show the following corollary:

\begin{corollary}\label{cor:3}
Let $\ket{\psi}$ be an entangled two qudit state. Then, the problem $\app$ is $\PQMA$-complete even when the Hamiltonian $H$ is restricted to be of the form
\[
	H=\sum_{j,k} \alpha_{j,k} \proj{\psi}_{j,k} ,
\]
where $\alpha_{j,k} \in \R$ and $\proj{\psi}_{j,k}$ denotes the projector onto $\ket{\psi}$ on qudits $j$ and $k$.
\end{corollary}

\noindent Each of these corollaries follows as the corresponding references show that the described families of Hamiltonians can efficiently simulate all spatially sparse Hamiltonians.

\subsection{The complexity of \app\ on the line}\label{sscn:results3}

We finally move to our last result, which characterizes the complexity of $\app$ on the line. Historically, it was known that the $\NP$-complete problem MAX-$2$-SAT on a line is efficiently solvable via dynamic programming or divide-and-conquer (even for large, but constant, dimension). It hence came as a surprise when \cite{AGIK09} showed that $\klh[2]$ on a line is still QMA-complete. This result was for local dimension $13$ (\cite{AGIK09} actually claimed a result for $12$-dimensional qudits; \cite{HNN13} later identified an error in their construction, and gave a fix requiring an addition dimension). \cite{Nag08} improved this to hardness for $12$-dimensional qudits by leveraging the parity of the position of qudits (similarly, \cite{Nag08} claimed a result for 11-dimensional particles, but suffered from the same error as \cite{AGIK09}). Most recently, \cite{HNN13} showed $\QMA$-completeness for qudits of dimension $8$ by allowing some of the clock transitions to be ambiguous (a similar idea was used in \cite{KKR06} to show $\QMA$-completeness of $\klh[2]$). The complexity of $\klh$ on a 1D line remains open for local dimension $2\leq d\leq 7$.

Returning to the setting of $\app$, it is clear that the classical analogue of $\app$ on a 1D line of bits is also in P; given any $2$-local Boolean formula $\phi:\set{0,1}^n\mapsto\set{0,1}$, we simply compute an optimal solution $x$ to $\phi$ (which recall can be done in 1D as referenced above), and subsequently evaluate any desired efficiently computable function on $x$ (i.e. a ``measurement'' on a subset of the bits). This raises the question: is $\app$ on a line still $\PQMA$-complete? Or does its complexity in the 1D setting drop to, say, $\QMA$? Our final result shows the former.

\begin{theorem}\label{thm:apxsim-1dline-complete}
	$\app$ is $\PQMA$-complete even when restricted to Hamiltonians on a 1D line of $8$-dimensional qudits and single-qudit observables.
\end{theorem}

\noindent Thus, even in severely restricted geometries like the 1D line, simulating a measurement on a single qudit of the ground space remains harder than QMA.\\

\noindent\emph{Proof techniques for Theorem~\ref{thm:apxsim-1dline-complete}.}
We employ a combination of new and known ideas. We wish to simulate the idea from~\cite{GY16} that instead of having the $\Poly$ machine make $m$ queries to a $\QMA$ oracle, it receives the answers to the queries as a ``proof'' $y\in\set{0,1}^m$ which it accesses whenever it needs a particular query answer. In~\cite{GY16}, Ambainis's query Hamiltonian~\cite{A14} was then used to ensure $y$ was correctly initialized. However, it is not clear how to use Ambainis' query Hamiltonian (or variants of it) while maintaining a 1D layout. We hence take a different approach.

Instead of receiving the query \emph{answers}, the $\Poly$ machine now has access to $m$ $\QMA$ verifiers $\set{V_i}_{i=1}^m$ corresponding to the $m$ queries, and for each of them receives a \emph{quantum} proof $\ket{\psi_i}$ in some proof register $R_i$. The $\Poly$ machine then treats the (probabilistic) outputs of each $V_i$ as the ``correct'' answer to the query $i$. If a query $i$ is a NO instance of a $\QMA$ problem, this works well --- no proof can cause $V_i$ to accept with high probability. However, if query $i$ is a YES instance, a cheating prover may nevertheless submit a ``bad'' proof to verifier $V_i$, since flipping the output bit of $V_i$ may cause the $\Poly$ machine to flip its final output bit. To prevent this, and thus ensure the $\Poly$ machine receives all correct answers with high probability, we use a delicate application of $1$-local energy penalties, which we call ``sifters'', to the outputs of the $V_i$; just enough to penalize bad proofs for YES cases, but not enough to cause genuine NO cases to incur large energy penalties. Here, we again utilize our result that $\Plog{QMA}=\Ppar{QMA}$ (Corollary~\ref{cor:complexityequalities}), and choose to begin with a $\PQMApar$ instance; this allows us to apply identical, independent sifters to the output of each verifier $V_i$, significantly easing the subsequent analysis and transition to 1D.

We next plug this construction, where the $\Poly$ circuit has many sub-circuits $V_i$, into the 1D $8$-dimensional circuit-to-Hamiltonian construction of~\cite{HNN13}.
Similarly to~\cite{GY16}, we apply a corollary of the Projection Lemma of~\cite{KKR06,GY16} (Corollary~\ref{cor:kkr}) to argue that any low energy state must be close to a history state $\ket{\psi}$. Combining with our sifter Hamiltonian terms, we show in Lemma~\ref{lem:L} that for $\ket{\psi}$ to remain in the low-energy space, it must encode $V_i$ outputting approximately the right query answer for any query $i$. To then conclude that all query responses are jointly correct with high probability, and thus that the low-energy space encodes the correct final output to the $\PQMApar$ computation, we apply a known quantum non-commutative generalization of the union bound.
In fact, our argument immediately shows hardness for both $\app$ and $\apptwo$.
The full proof is given in Sections~\ref{ssec:1d-construction},~\ref{ssec:1D-correctness}, and~\ref{sscn:lemmas}.

\subsection{Open questions and organization}
Our results bring previous $\PQMA$-hardness results for a remarkably natural problem, Approximate Simulation ($\app$), closer to the types of problems studied in the physics literature, where typically observables are $O(1)$-local, allowed interactions physically motivated, and the geometry of the interaction graph is constrained. There are many questions which remain open, of which we list a few here: (1) The coupling strengths for local Hamiltonian terms in Corollary~\ref{cor:1},\ref{cor:2},\ref{cor:3} are typically non-constant, as these corollaries follow from the use of existing perturbation theory gadgets; can these coupling constants be made $O(1)$? Note this question is also open for the complexity classification of $\klh$ itself~\cite{CM16,PM17}. (2) What is the complexity of $\PQMA$? It is known that $\PQMA\subseteq \class{PP}$~\cite{GY16}; can a tighter characterization be obtained? (3) Can similar hardness results for $\app$ be shown for \emph{translationally invariant} 1D systems? For reference, it is known that $k$-LH is $\class{QMA}_{\class{exp}}$-complete for 1D translationally invariant systems when the local dimension is roughly $40$~\cite{GI13,BCO17}. ($\class{QMA}_{\class{exp}}$ is roughly the quantum analogue of NEXP, in which the proof and verification circuit are exponentially large in the input size. The use of this class is necessary in~\cite{GI13,BCO17}, as the only input parameter for 1D translationally invariant systems is the \emph{length} of the chain.) If a similar hardness result holds for $\app$, presumably it would show $\class{P}^{\class{QMA}_{\class{exp}}[\log]}$-hardness for 1D translationally invariant systems.

\paragraph{Organization.} We introduce notation and definitions in Section~\ref{sec:preliminaries}.
We prove that $\app$ is contained in $\Plog{C}$ in Section~\ref{ssec:containment} for classes $\class{C}$ and corresponding restrictions, and that $\apptwo$ is $\Ppar{C}$-hard in Section~\ref{ssec:hardness}, thereby proving Theorem~\ref{thm:APXcompleteness}.
In Section~\ref{sec:simulations}, we introduce a special case of the definition of simulation from \cite{CMP18} and show that simulations correspond to reductions of the problem $\apptwo$, yielding Theorem~\ref{thm:S-APXSIM}; proofs with regard to the general definition are in Appendix~\ref{sec:generalsimulations}.
In Section~\ref{sec:spatiallysparse}, we give a spatially sparse construction with which $\apptwo$ is $\PQMApar$-hard, thus proving Theorem~\ref{thm:simofsparsehard}. Finally, in Section~\ref{sec:1D}, we study hardness on a 1D line and prove Theorem~\ref{thm:apxsim-1dline-complete}.

\section{Preliminaries}
\label{sec:preliminaries}
\paragraph{Notation.}
Let $\lambda(H)$ denote the smallest eigenvalue of Hermitian operator $H$.
For a matrix $A$, $\snorm{A} := \max\{\norm{A\ket{v}}_2 : \norm{\ket{v}}_2 = 1\}$ is the operator norm or spectral norm of $A$, and $\trnorm{A}:=\trace{\sqrt{A^\dagger A}}$ the trace norm. Throughout this paper, we will assume generally that both $H=\sum_{i=1}^m H_i$ and observable $A=\sum_{i=1}^m A_i$ are local Hamiltonians whose local terms $H_i$ and $A_i$ act non-trivially on at most $O(\log n)$ out of $n$ qubits. We also assume $m,\snorm{H_i},\snorm{A_i}\in O(\poly~n)$ for all $i\in\set{1,\ldots,m}$, for $n$ the number of qubits in the system. For a subspace $S$, $S^{\perp}$ denotes the orthogonal complement of $S$. We denote the restriction of an operator $H$ to subspace $S$ as $H\vert_S$. The null space of $H$ is denoted $\Null(H)$.

\paragraph{Definitions.}
$\PQMA$, defined in~\cite{A14}, is the set of decision problems decidable by a polynomial-time deterministic Turing machine with the ability to query an oracle for a $\QMA$-complete problem $O(\log n)$ times, where $n$ is the size of the input.
For a class $\class{C}$ of languages or promise problems, the class $\Plog{C}$ is similarly defined, except with an oracle for a $\class{C}$-complete problem.

$\Ppar{C}$ is the class of problems decidable by a polynomial-time deterministic Turing machine given access to an oracle for a $\class{C}$-complete problem, with the restriction that all (up to $O(n^c)$ for $c\in\Theta(1)$) queries to the oracle must be made in one time step, i.e. in parallel. Such queries are labeled \emph{non-adaptive}, as opposed to the \emph{adaptive} queries allowed to a $\Plog{C}$ machine.
We note that $\PNPpar$ has in the past been denoted $\leq^{p}_{tt}(\NP)$, in reference to polynomial-time truth-table reductions (e.g. \cite{BH91}).

In this article, for $\PQMA$ we assume oracle queries made by the P machine are to an oracle for the QMA-complete~\cite{KSV02} $k$-local Hamiltonian problem ($\klh$), defined as follows: Given a $k$-local Hamiltonian $H$ and inverse polynomial-separated thresholds $a,b\in\reals$, decide whether $\lambda(H)\leq a$ (YES-instance) or $\lambda(H)\geq b$ (NO-instance) \cite{KKR06}. We shall say an oracle query is \emph{valid} (\emph{invalid}) if it satisfies (violates) the promise gap of the $\QMA$-complete problem the oracle answers. (An invalid query hence satisfies $\lambda(H)\in (a, b)$.) For any invalid query, the oracle can accept or reject arbitrarily. A \emph{correct} query string $y\in\set{0,1}^m$ encodes a sequence of correct answers to all of the $m$ queries made by the $\class{P}$ machine, and an \emph{incorrect} query string is one which contains at least one incorrect query answer. Note that for an invalid query, any answer is considered ``correct'', yielding the possible existence of multiple correct query strings. Nevertheless, the \class{P} machine is required to output the same final answer (accept or reject) regardless of how such invalid queries are answered \cite{G06}. The above definitions extend analogously when the class $\QMA$ is replaced with another class $\class{C}$, with a designated $\class{C}$-complete problem $\Pi_C$ playing the role of $\klh$. (In this paper, the complexity classes $\class{C}$ we consider have complete problems.)

%=========================================================================================
%  SECTION:
%=========================================================================================
\section{Parallel versus adaptive queries}\label{sec:parallel}
We begin by showing Theorem~\ref{thm:APXcompleteness}, i.e. that $\Plog{C}=\Ppar{C}$ for appropriate complexity classes $\class{C}$. Section~\ref{ssec:containment} shows containment of the corresponding $\app$ problem in $\Plog{C}$ (and thus in $\Ppar{C}$). Section~\ref{ssec:hardness} then shows $\Ppar{C}$-hardness (and thus $\Plog{C}$-hardness) of $\app$. Theorem~\ref{thm:APXcompleteness} is restated and proven in Section~\ref{ssec:final}.

\subsection{Containment in $\Plog{C}$}
\label{ssec:containment}

We begin by modifying the containment proof of~\cite{A14} to show containment of \app\ in classes $\Plog{C}$ for $\class{C}$ beyond just $\class{C}=\class{QMA}$.

\begin{lemma}
	\label{lem:APXinPXlog}
	Let $H$ be a $k$-local Hamiltonian acting on $n$ qudits, and let $A$ be an observable on the same system of $n$ qudits. If $\klh$ for $\alpha H+ \beta A$ is contained in complexity class $\class{C}$ for any $0\le \alpha, \beta \le \poly(n)$ and for all $k\geq 1$, then $\app(H,A,k,\ell,a,b,\delta)\in\Plog{C}$ for all $\ell\leq O(\log n)$ and $b-a,\delta\geq O(1/\poly~n)$.
\end{lemma}

\begin{proof}
	We need to show the existence of a $\poly(n)$ time classical algorithm to decide $\app$ while making at most  $O(\log n )$ queries to an oracle for $\class{C}$. As with the proof in \cite{A14}, the idea is to use $O(\log n )$ oracle queries to determine the ground space energy $\lambda(H)$ of $H$ by binary search, and then use one final query to determine the answer.
	In \cite{A14} the final query is a $\QMA$ query; here we show how this final query can be performed differently so that only an oracle for $\class{C}$ is required.
	
	First calculate a lower bound $\mu$ for $\lambda(A)$, the lowest eigenvalue of $A$. If $A$ acts only on $O(1)$ qudits, then $\lambda(A)$ can be calculated via brute force (up to, say, inverse exponential additive error) in $O(1)$ time. If $A$ acts on many qudits, then $\lambda(A)$ can alternatively be approximated to within inverse polynomial additive error by binary search (as in~\cite{A14}) by querying the $\class{C}$ oracle $O(\log \|A\|)=O(\log n)$ times. Note that without loss of generality, we may assume $0\leq b-\mu \le q(n)$ for some efficiently computable polynomial $q$. The lower bound holds since if $b<\mu\leq\lambda(A)$, we conclude our \app\ instance is a NO instance, and we reject. For the upper bound, it holds that $\mu\leq \snorm{A}$, and we may assume $b\leq \snorm{A}$, as otherwise our \app\ instance is either a YES or invalid instance, and in both cases we can accept. By assumption, $\snorm{A}\leq q(n)$ for appropriate polynomial $q$ which can be computed efficiently by applying the triangle inequality to the local terms of $A$; note $\snorm{A}$ may hence be replaced by $q$ in the bounds above.
	
	Perform binary search with the oracle for $\class{C}$ (an example of how to perform binary search with an oracle for a promise problem is given in~\cite{A14}) to find $\lambda^*$ such that $\lambda(H) \in [\lambda^*,\lambda^*+\epsilon]$ where
	\[
        \epsilon=\frac{\delta(b-a)}{2(b-\mu)} \ge 1/\poly(n)
    \]
	since $0\leq b-\mu \le \poly (n)$.
	This requires $O(\log 1/\epsilon)=O(\log n)$ queries to the oracle for $\class{C}$.
	Next perform one final query to the $\class{C}$ oracle to solve \klh\ with Hamiltonian $H'$ with thresholds $a'$ and $b'$, where
	\[
        H'=(b-\mu)H+\delta A  \quad \text{ and } \quad  \begin{array}{l}
	   a'=(\lambda^*+\epsilon) (b-\mu) +\delta a\\
	   b'=\lambda^* (b-\mu) +\delta b
	\end{array}\]
	and accept if and only if this final query accepts. Observe this is an allowed query for the $\class{C}$ oracle because $H'$ is of the form required in the statement of the lemma (recall $b-\mu\geq 0$), and also
	\[
        b'-a'=\delta(b-a)-\epsilon(b-\mu)=\delta(b-a)/2 \ge 1/\poly( n).
	\]	
	Now, if $\app(H,A,k,l,a,b,\delta)$ is a YES instance, then there exists $\ket{\psi}$ such that $\bra{\psi} H\ket{\psi}=\lambda(H)$ and $\bra{\psi} A \ket{\psi} \le a$. Then
	\[
        \bra{\psi} (b-\mu)H + \delta A \ket {\psi} \le \lambda(H) (b-\mu) +\delta a \le (\lambda^*+\epsilon) (b-\mu) +\delta a =a'
    \]
	and the algorithm accepts as required.
	
	Now suppose the input is a NO instance.
	We will show that $\bra{\psi}H'\ket{\psi} \ge b'$ for any $\ket{\psi}$ and so the algorithm rejects as required.
	First, if $\ket{\psi}$ is low-energy with $\bra{\psi} H\ket{\psi} \le \lambda(H) + \delta$, then it also satisfies $\bra{\psi} A \ket{\psi} \ge b$, and so
	\[\bra{\psi} (b-\mu)H + \delta A \ket {\psi} \ge \lambda(H)(b-\mu) +\delta b \ge \lambda^* (b-\mu) +\delta b=b'  \]
	where we have used $\bra{\psi}H\ket{\psi} \ge \lambda(H) \ge \lambda^*$ and $b-\mu\geq 0$. Otherwise, if $\ket{\psi}$ is high energy with $\bra{\psi} H\ket{\psi}\ge \lambda(H) + \delta$, then
	\begin{align*}
	\bra{\psi} (b-\mu)H + \delta A \ket {\psi} &\ge (\lambda(H)+\delta) (b-\mu) +\delta \lambda(A) \\&= \lambda(H) (b-\mu) +\delta b+\delta(\lambda(A)-\mu) \ge \lambda^* (b-\mu) + \delta b=b' \end{align*}
	where we have used $\bra{\psi} A \ket{\psi} \ge \lambda(A)$ and $ \lambda(A) -\mu \ge 0$. Thus, we reject.
\end{proof}

An additional application of Lemma~\ref{lem:APXinPXlog} is that it allows us to prove that the $\app$ problem is easy for certain families of Hamiltonians for which $\klh$ is known to be easy. For example, the work on ferromagnetic Hamiltonians in \cite{BG17} implies the following corollary:

\begin{corollary}
	Consider the family of Hamiltonians $\mathcal{F}$ of the form:
	\begin{equation}
	\label{eq:ferromagnetic} H=\sum_{1\leq i < j \leq n} (-b_{ij}X_iX_j +c_{ij}Y_iY_j)+\sum_{i=1}^n d_i (I+Z_i)
	\end{equation}
	where $b_{ij}, c_{ij}, d_i \in \R$ satisfy $|c_{ij}|\le b_{ij}$.
	Then, $\app$ for Hamiltonians and observables chosen from $\mathcal{F}$ is contained in $\class{BPP}$.
\end{corollary}

\begin{proof}
	In \cite{BG17}, it was shown that for Hamiltonians in $\mathcal{F}$, there exists a FPRAS (fully polynomial randomized approximation scheme) to calculate the partition function of $H$ up to multiplicative error. In particular, it is noted that this gives a corresponding approximation to the ground state energy with additive error. Therefore, there exists a randomized algorithm that runs in polynomial-time and which, with high probability, gives an approximation to the ground state energy of $H$ up to inverse-polynomial additive error. This algorithm shows containment of $\klh$ restricted to the family $\mathcal{F}$ in $\BPP$.
	
	We now wish to consider $\app$ by applying Lemma~\ref{lem:APXinPXlog}, but first need to check that $H'=\alpha H+ \beta A$ is in the family for all $\alpha,\beta \ge 0$.
	It is clear that $H'$ can be written in the form of Equation~\eqref{eq:ferromagnetic}, but not whether it satisfies the required bounds on its coefficients.
	Following the notation of Equation~\eqref{eq:ferromagnetic}, for an operator $F$ in the family $\mathcal{F}$, let $c_{i,j}(F)$ be the coefficient of $Y_iY_j$ and let $b_{i,j}(F)$ be the coefficient of $-X_iX_j$.
	Then, a simple application of the triangle inequality shows that
	\[ |c_{i,j}(H')|=| \lambda c_{i,j}(H)+\mu c_{i,j}(A) | \le \lambda |c_{i,j}(H)_{i,j}|+\mu |c_{i,j}(A)| \le \lambda b_{i,j}(H) + \mu b_{i,j}(A) =b_{i,j}(H').\]
	Therefore, by Lemma~\ref{lem:APXinPXlog}, $\app$ is contained in $\Plog{\BPP}$ for $H$ and $A$ in $\mathcal{F}$.
	Finally, we note that $\Plog{BPP}=\BPP$, since clearly $\BPP \subseteq \Plog{BPP}$, and $\Plog{BPP} \subseteq \Poly^{\BPP} \subseteq \BPP^{\BPP} = \BPP$ since $\BPP$ is low for itself.
\end{proof}

%=========================================================================================
%  SECTION: HARDNESS PROOF
%=========================================================================================
\subsection{Hardness for $\Ppar{C}$}
\label{ssec:hardness}

We next modify the proof that $\app$ is $\PQMA$-hard to obtain the following lemma. Our modifications include simplifying the ``query Hamiltonian'' of~\cite{A14} and improving the construction of~\cite{GY16} by using the Cook-Levin theorem, as opposed to Kitaev's circuit-to-Hamiltonian construction. The latter has a nice consequence --- in contrast to the $\QMA$-completeness results for $\klh$, where the promise gap is inverse polynomial, for $\app$ we are able to show that the promise gap $b-a$ sufficient for $\Ppar{C}$-completeness scales as $\Omega(1)$.

\begin{lemma}
	\label{lem:APXisPparXhard}
	Let $\cF$ be a family of Hamiltonians for which \klh\ is $\class{C}$-hard for all $k\geq 2$. Then $\apptwo$ is $\Ppar{C}$-hard even when $b-a=\Omega(1)$, the observable $A$ is a  single Pauli $Z$ measurement, and when restricted to Hamiltonians of the form $H=H_{\text{cl}}+\sum_i \proj{1}_i\otimes H_i$, where $H_{\text{cl}}$ is a classical Hamiltonian, and the $H_i$ are Hamiltonians from $\cF$.
\end{lemma}

\noindent To show this, we require two tools: in the next two subsections, we show how to simplify~\cite{A14}'s query Hamiltonian in the context of parallel queries, used to enforce correct query answers, and discuss how to employ the Cook-Levin reduction, which enforces a correct simulation of the circuit given those query answers, respectively.

\subsubsection{Simplifying Ambainis' query Hamiltonian}

First, we give a simplified version of the ``query Hamiltonian'' introduced by Ambainis~\cite{A14}, which will be useful in the following lemmas.
We note that~\cite{GY16} reduced the locality of the construction of~\cite{A14} by applying the unary encoding trick of Kitaev~\cite{KSV02}, but due to the simplified structure of parallel queries, here we do not require this unary encoding to achieve $O(1)$-locality for our Hamiltonian. However, \cite{GY16} also reduced the locality of the observable from $O(\log n)$-local to a single qubit, by deferring the job of simulating the circuit away from the observable and to the Hamiltonian from~\cite{KSV02}, and this improvement is now crucial, as otherwise a polynomial number of queries would demand an $O(\poly~n)$-local observable.

Given some $\Ppar{C}$ computation $U$ for an appropriate class $C$, let $(H_{\spa{Y}_i},a_i,b_i)$ be the instance of (without loss of generality) $\klh[2]$ corresponding to the $i$-th query made by $U$. Then, our ``query Hamiltonian'' is
\begin{equation}\label{eqn:Hdef}
H = \sum_{i=1}^{m} M_i := \sum_{i=1}^{m} \left(\frac{a_i+b_i}{2} \ketbra{0}{0}_{\spa{X}_{i}}\otimes I_{\spa{Y}_i} + \ketbra{1}{1}_{\spa{X}_{i}}\otimes H_{\spa{Y}_i} \right) ,
\end{equation}
where single qubit register $\spa{X}_i$ is intended to encode the answer to query $i$ and $\spa{Y}_i$ encodes the ground state of $H_{\spa{Y}_i}$. Since each query is $2$-local, $H$ is $3$-local. Notably, because $U$ makes all of its queries in \emph{parallel}, we are able to weight each of the $m$ terms equally, unlike in~\cite{A14,GY16} which studied adaptive queries. This significantly eases our later analysis.

The following lemma is analogous to Lemma~3.1 of \cite{GY16}, but with an improved spectral gap. The proof is similar to theirs, but is significantly simplified due to our use of parallel queries.

\begin{lemma}\label{lem:hgap}
	Define for any $x\in\set{0,1}^m$ the space $\spa{H}_{x_1\cdots x_m} := \bigotimes_{i=1}^m \ketbra{x_i}{x_i}\otimes \spa{Y}_i $. Then, there exists a correct query string $x\in\set{0,1}^m$ such that the ground state of $H$ lies in $\spa{H}_{x_1\cdots x_m}$. Moreover, if $\lambda$ is the minimum eigenvalue of $H$ restricted to this space, then for any incorrect query string $y_1\cdots y_m$, any state in $\spa{H}_{y_1\cdots y_m}$ has energy at least $\lambda+\epsilon$, where $\epsilon =\min_i (b_i-a_i)/2$.
\end{lemma}
\begin{proof}
	We proceed by contradiction. Let $x\in\set{0,1}^m$ ($y\in\set{0,1}^m$) denote a correct (incorrect) query string which has lowest energy among all \emph{correct} (\emph{incorrect}) query strings against $H$. (Note that $x$ and $y$ are well-defined, though they may not be unique; in this latter case, any such $x$ and $y$ will suffice for our proof.) For any $z\in\set{0,1}^m$, define $\lambda_z$ as the smallest eigenvalue in $\spa{H}_z$.
	
	Since $y$ is an incorrect query string, there exists at least one $i\in\{1,\dots,m\}$ such that $y_i$ is the wrong answer to a valid query $H_{\spa{Y}_i}$. If query $i$ is a YES-instance, the smallest eigenvalue of $M_i$ corresponds to setting $\spa{X}_i$ to (the correct query answer) $\ket{1}$, and is at most $a_i$. On the other hand, the space with $\spa{X}_i$ set to $\ket{0}$ has all eigenvalues equaling $(a_i+b_i)/2$. A similar argument shows that in the NO-case, the $\ket{0}$-space has eigenvalues equaling $(a_i+b_i)/2$, and the $\ket{1}$-space has eigenvalues at least $b_i$. We conclude that flipping query bit $i$ to the correct query answer $\overline{y}_i$ allows us to ``save'' an energy penalty of $(b_i-a_i)/2$ against $M_i$, and since all other terms act invariantly on $\spa{X}_i\otimes\spa{Y}_i$, we save $(b_i-a_i)/2$ against $H$ as well.
	
	Let $y'$ denote $y$ with bit $i$ flipped.
	If $y'$ is also an incorrect query string, we have $\lambda_{y'}<\lambda_y$, a contradiction due to the minimality of $y$. Conversely, if $y'$ is a correct query string, then we must have $\lambda_{y'}\geq \lambda_x+(b_i-a_i)/2\ge \lambda +\epsilon$, as otherwise we contradict the minimality of $x$.
\end{proof}

\subsubsection{Cook-Levin construction}
\label{sec:cooklevin}

\begin{figure}[h]
	\centering
	\begin{tikzpicture}
	\begin{scope}[shift={(7,0)}]
	\foreach \x in {1,...,5}{
		\foreach \y in {1,...,4}{
			\node[circle, fill=black] at (\x,\y){};
	}}
	\draw (1.7,0.7) rectangle (4.3,2.3);
	\draw (0.7,1.7) rectangle (2.3,3.3);
	\draw (3.7,2.7) rectangle (5.3,4.3);
	\node at (3,1.5) {$h_{U_1}$};
	\node at (1.5,2.5) {$h_{U_2}$};
	\node at (4.5,3.5) {$h_{U_3}$};
	
	\draw (1,1) -- (1,2);
	\draw (5,1) -- (5,2);
	\draw (3,2) -- (3,3);
	\draw (4,2) -- (4,3);
	\draw (5,2) -- (5,3);
	\draw (1,3) -- (1,4);
	\draw (2,3) -- (2,4);
	\draw (3,3) -- (3,4);
	
	\node at (6.5,1) {$t=1$};
	\node at (6.5,2) {$t=2$};
	\node at (6.5,3) {$t=3$};
	\node at (6.5,4) {$t=4$};
	\end{scope}
	
	\foreach \x in {1,...,5}{
		\draw (\x,0.7) - -(\x,4.3);
	}
	\draw[fill=white] (1.7,1) rectangle (4.3,2);
	\draw[fill=white] (0.7,2.3) rectangle (2.3,3.3);
	\draw[fill=white] (3.7,3) rectangle (5.3,4);
	
	\node at (3,1.5) {$U_1$};
	\node at (1.5,2.8) {$U_2$};
	\node at (4.5,3.5) {$U_3$};
	\end{tikzpicture}
	\caption[Cook-Levin construction]{Cook-Levin construction of classical Hamiltonian to simulate a $\class{P}$ machine. On the left is a picture of the gates $U_i$ in the circuit of the $\Poly$ machine; the figure on the right shows the Hamiltonian terms $h_{U_t}$ encoding each gate. Each straight line edge on the right represents the interaction $\proj{01}+\proj{10}$. The initialization terms $H_{\text{in}}$ on qubits in time step $t=0$ are omitted in the diagram.}
	\label{fig:cooklevin}
\end{figure}
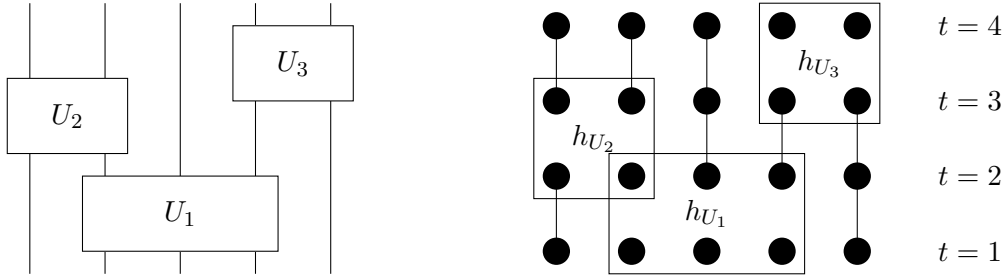

We now show how to model the Cook-Levin construction as a Hamiltonian in our setting. For this, we consider the $\class{P}$ machine to be given as a circuit of classical reversible gates $U=U_m\dots U_1$, in which one gate occurs at each time step.
The evolution of the circuit is encoded into a 2D grid of qubits, where the $t$-th row of qubits corresponds to the state of the system at time step $t$; the output of the circuit is copied to a dedicated output bit in the final timestep. The overall Hamiltonian is diagonal in the computational basis with a groundspace of states corresponding to the correct evolution of the $\class{P}$ machine.

Let $I_t$ be the set of qubits which $U_t$ acts non-trivially on. If a qubit $i \notin I_t$ (i.e. it is not acted on by the circuit at time step $t$), then there is an interaction $\proj{01}+\proj{10}$ on qubits $(i,t)$ and $(i,t+1)$, to penalize states which encode a change on qubit $i$. To encode a classical reversible gate $U_t: x \mapsto U_t(x)$ acting at time $t$, we define an interaction $h_{U_t}=I-\sum_x\proj{x}_t \otimes \proj{U_t(x)}_{t+1}$ acting non-trivially only on qubits $(i,t')$ for $i \in I_t$ and $t'$ equal to $t$ or $t+1$.
See Figure~\ref{fig:cooklevin} for a pictorial representation of this Hamiltonian.
Then
\begin{equation}
\label{eqn:Hprop}
H_{\operatorname{prop}}=\sum_{t=1}^m \left(h_{U_t}+\sum_{i \notin I_t}\proj{0}_{(i,t)}\proj{1}_{(i,t+1)}+\proj{1}_{(i,t)}\proj{0}_{(i,t+1)}\right)
\end{equation}
 is positive semi-definite and has ground space spanned by states of the form:

\[\ket{w(x)}=\ket{x}_{t=1} \otimes \ket{U_1 x}_{t=2} \otimes \dots \otimes \ket{U_m\dots U_1 x}_{t=m+1} \]

\noindent Typically, there is an additional term $H_{\text{in}}$ consisting of 1-local $\proj{1}$ terms on all qubits in the first ($t=1$) row.
Then the  Hamiltonian $H_{\text{prop}}+H_{\text{in}}$ has (1) unique ground state $\ket{w(0^n)}$ encoding the action of the circuit on the $0^n$ string, (2) ground state energy $0$, and (3) spectral gap at least $1$, since the Hamiltonian is a sum of projectors. We will later show how we adapt $\hin$ to our query answer register.

\subsubsection{Proof of hardness}

We are almost ready to prove the main result of this section, Lemma~\ref{lem:APXisPparXhard}. Before doing so, we require a final technical lemma.

\begin{lemma}
	\label{lem:lowenergydistance}
	Let $H$ be a Hamiltonian and $\rho$ a density matrix satisfying $\tr(H\rho) \le \lambda(H) +\delta$. Let $P$ be the projector onto the space of eigenvectors of $H$ with energy less than $\lambda(H) + \delta'$. Then,
	\[
	\frac{1}{2}\|\rho-\rho'\|_1 \le \sqrt{\frac{\delta}{\delta'}} , \quad \text{ where } \rho'=P \rho P / \tr(P\rho) .
	\]
	
\end{lemma}

\begin{proof}
	First, bound the trace distance by the fidelity in the usual way (using one of the Fuchs-van de Graf inequalities \cite{FvdG99}):
	\begin{equation}
	\label{eq:tracefidelity}
	\frac{1}{2}\|\rho-\rho'\|_1 \le \sqrt{1-F(\rho,\rho')^2}
	\end{equation}
	where
	\[
	F(\rho,\rho')
	=\tr\left(\sqrt{\sqrt{\rho} \rho'\sqrt{\rho}}\right)
	=\tr\left(\sqrt{\frac{\sqrt{\rho}P\rho P\sqrt{\rho}}{\tr(P\rho)}}\right)
	=\frac{1}{\sqrt{\tr(P\rho)}} \tr(\sqrt{\rho}P\sqrt{\rho})
	=\sqrt{\tr(P\rho)},
	\]
	where the third equality follows since $(\sqrt{\rho}P\sqrt{\rho})^2=\sqrt{\rho}P\rho P\sqrt{\rho}$ and since the latter is positive semi-definite. Now, it remains to bound $\tr(P\rho)$.
	We note that $H$ has eigenvalues at least $\lambda(H)+\delta'$ on the space annihilated by $P$ and eigenvalues at least $\lambda(H)$ everywhere else, and so $H \succeq (\lambda(H)+\delta')(I-P) +\lambda(H) P=(\lambda(H)+\delta')I-\delta' P$. Therefore, using the bound on $\tr(H\rho)$, we have
	\[
	\lambda(H)+\delta\ge \tr(H\rho)\ge (\lambda(H)+\delta')\tr(\rho)-\delta'\tr(P\rho) \quad \Leftrightarrow \quad 1-\tr(P\rho) \le \frac{\delta}{\delta'} .
	\]
	Substituting this back into Equation \eqref{eq:tracefidelity} proves the result.
\end{proof}

We are now ready to prove Lemma~\ref{lem:APXisPparXhard}:
\begin{proof}[Proof of Lemma~\ref{lem:APXisPparXhard}]
	We split the Hilbert space into three parts $\W$, $\X= \bigotimes_i\X_i$, $\Y= \bigotimes_i\Y_i$ and have a Hamiltonian of the form $H=H_1+H_2$, where $H_1$ acts on $\W$ and $\X$, and $H_2$ acts on $\X$ and $\Y$. $H_2$ is the query Hamiltonian of Equation (\ref{eqn:Hdef}), and therefore by Lemma~\ref{lem:hgap} the space of eigenvectors of $H_2$ with eigenvalues less than $\lambda(H_2)+\epsilon$ is spanned by states of the form: $\ket{x}_{\X} \otimes \ket{\phi}_{\Y}$, where $x$ is a correct string of answers for the queries to the $\class{C}$ oracle.

 $H_1=H_{\text{prop}}+H_{\text{in}}$ is the classical Hamiltonian encoding the evolution of a classical $\Poly$ circuit, using the Cook-Levin construction of Section~\ref{sec:cooklevin}, where $H_{\text{prop}}$ is as defined in Equation (\ref{eqn:Hprop}). For clarity, $H_{\text{prop}}$ and $H_{\text{in}}$ act on $\W$ and $\W\otimes\X$, respectively. We think of $\W$ as ``laid out in a 2D grid'' as in Figure~\ref{fig:cooklevin}, and of $\X$ as playing the role of a ``message'' register passing information between $H_1$ and $H_2$. We modify the Hamiltonian $H_{\text{in}}$ which initializes the qubits at the start of the classical circuit. For each qubit $\X_i$  in $\X$, we initialize a corresponding qubit of the first ($t=0$) row of $\W$ into the same state with a penalty term $\proj{1}_{\X_i}\otimes\proj{0}_{\W_i}+\proj{0}_{\X_i}\otimes\proj{1}_{\W_i}$. All other qubits in the first ($t=0$) row of $\W$ are initialized to $\ket{0}$ with a penalty $\proj{1}$. The full construction is depicted diagrammatically in Figure~\ref{fig:prevstructure}. Note that as stated in the claim, $H$ is of the form $H=H_{\text{cl}}+\sum^{m}_i \proj{1}_i\otimes H_i$, where $H_{\text{cl}}$ contains $H_1$ and the local terms of $H_2$ which are tagged with $\ketbra{0}{0}$ in registers $\X_i$.

We can argue about the low-energy eigenspace of $H$ as follows. Since the ground spaces of $H_1$ and $H_2$ have non-trivial intersection, $\lambda(H)=\lambda(H_1)+\lambda(H_2)=\lambda(H_2)$. Moreover, since $[H_1,H_2]=0$ (they overlap only on the $\X$ register, on which they are both diagonal in the standard basis), and since we may assume without loss of generality that $\lambda(H_2)+\epsilon$ is inverse polynomially bounded below $1$ (otherwise, we can scale $H_1$ by an appropriate fixed polynomial), we conclude the space of eigenstates of $H$ with eigenvalue less than $\lambda(H)+\epsilon$, henceforth denoted $\spa{H}_{\rm low}$, is spanned by states of the form $\ket{\Phi}=\ket{w}_{\W}\otimes \ket{x}_{\X} \otimes \ket{\phi}_{\Y}$, where $x$ is a string of correct answers to the oracle queries and $w$ is the classical string encoding the correct computation of the $\class{P}$ circuit acting on $x$. The qubit corresponding to the output bit of the $\Poly$ circuit will be in the state $\ket{1}$ (resp. $\ket{0}$) in a YES (resp. NO) instance of $\apptwo$.

To complete the proof let the observable $A=Z_{\operatorname{out}}$, a Pauli $Z$ measurement on the qubit corresponding to the output bit of the $\Poly$ circuit, and let $\delta=\epsilon/16$ and $\delta'=\epsilon$. Consider any state $\ket{\psi}$ with $\bra{\psi}H\ket{\psi} \le \lambda(H) + \delta$. Then by Lemma~\ref{lem:lowenergydistance}, there exists a state $\ket{\psi'}\in\spa{H}_{\rm low}$ such that $\bra{\psi'}H\ket{\psi'}\leq \lambda(H)+\delta'=\lambda(H)+\epsilon$ which satisfies $\|\ketbra{\psi}{\psi}-\ketbra{\psi'}{\psi'}\|_1 \le 1/2$.
So,
\[\bra{\psi'}Z_{\operatorname{out}}\ket{\psi'}=\left\{\begin{array}{cl} -1 &\text{ in a YES instance}\\ 1 &\text{ in a NO instance} \end{array} \right.\]
which implies by H\"{o}lder's inequality that $\bra{\psi}A\ket{\psi}$ is $\le -1/2$ in a YES instance and $\ge 1/2$ in a NO instance, as required.
\end{proof}
\begin{figure}
	\begin{tikzpicture}
	\foreach \a in {1,2,3}{
		\begin{scope}[shift={(\a*6,0)}]
		\node  at (3,0.5) {$\mathcal{Y}_{\a}$};
		\draw[green] (1,0) -- (1,1) -- ( 5,1) -- (5,0) -- ( 1,0);
		\node[circle,draw=blue] (\a) at (3,2) {$\mathcal{X}_{\a}$};
		\draw (3,1) -- (\a);
		\end{scope}
		\draw (\a) -- (15,3);
	}
	\draw[red] (12,3) rectangle (18,4);
	\node at (15,3.5) {$\W$};
	\end{tikzpicture}
	\caption[Structure of the Hamiltonian used in Lemma~\ref{lem:APXisPparXhard}]{The structure of the Hamiltonian $H=H_1+H_2$ used in Lemma~\ref{lem:APXisPparXhard}, for the case of $3$ queries. $H_1$ acts on the space $\W \otimes \X$ and $H_2$ acts on $\X \otimes \Y$, where $\X= \bigotimes_i \X_i$ and $\Y= \bigotimes_i \Y_i$}
	\label{fig:prevstructure}
\end{figure}
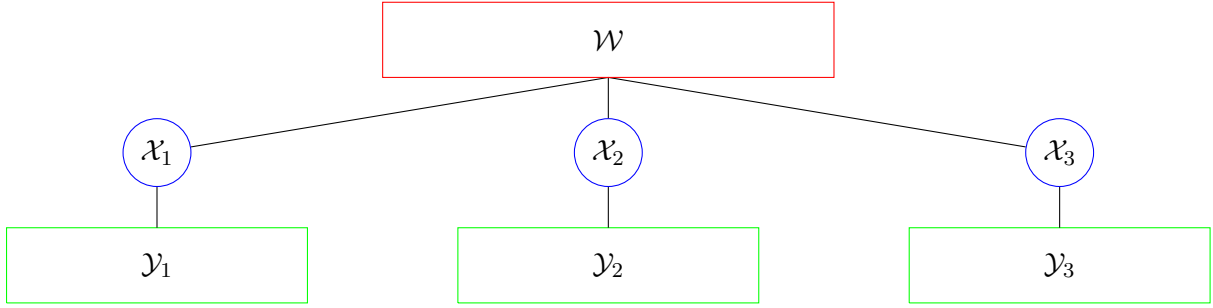

\subsection{Final result}\label{ssec:final}

Theorem~\ref{thm:APXcompleteness} is now a straightforward consequence of Lemma~\ref{lem:APXinPXlog} and Lemma~\ref{lem:APXisPparXhard}.
\begin{reptheorem}{thm:APXcompleteness}
	Let $\class{C}$ be a class of languages or promise problems. Let $\cF$ be a family of Hamiltonians for which \klh is $\class{C}$-complete under poly-time many-one reductions for all $k \ge 2$. Suppose $\cF$ is closed under positive linear combination of Hamiltonians, and that if $\{H_i\}_{i=1}^m\subset \cF$, then $H_{\text{cl}}+\sum_{i=1}^m \proj{1}_i\otimes H_i \in \cF$, where $H_{\text{cl}}$ is any classical Hamiltonian (i.e. diagonal in the standard basis). Then,
	\[
	\Plog{C}=\Ppar{C},
	\]
	and $\app$ is $\Plog{C}$-complete when restricted to $k$-local Hamiltonians and observables from $\cF$.
\end{reptheorem}

\begin{proof}
	The containment $\Poly^{\class{C[log]}} \subseteq \Poly^{||\class{C}}$ follows directly from the same argument that $\PNPlog\subseteq\PNPpar$ of \cite{Bei91}, which we summarized in Section~\ref{ssec:parVsAdap}. By Lemma~\ref{lem:APXinPXlog}, $\app$ is contained in $\Plog{C}$ for Hamiltonians and observables from $\cF$.
	And by Lemma~\ref{lem:APXisPparXhard} $\apptwo$ is $\Ppar{C}$-hard for Hamiltonians from $\cF$, even when the observable is a single Pauli $Z$ measurement, which is contained in $\cF$ by the assumption that $\cF$ contains any classical Hamiltonian $H_{\text{cl}}$.
	Since $\apptwo$ trivially reduces to $\app$, we thus have that $\app$ is similarly $\Ppar{C}$-hard, and the result follows.
\end{proof}

%=========================================================================================
%  SECTION: SIMULATIONS
%=========================================================================================

\section{Simulations and $\app$ for physical classes of Hamiltonians}
\label{sec:simulations}
In order to study the complexity of \app\ for physically motivated Hamiltonians in Section~\ref{sec:spatiallysparse}, we require two tools: first, hardness results for parallel query classes $\Ppar{C}$, given in Section~\ref{sec:parallel}, and second, an understanding of how \emph{simulations} affect the hardness of the problem $\app$, which this section focuses on. Specifically, we consider a simplified notion of simulation, defined below, which is a special case of the full definition given in \cite{CMP18}. This simpler case includes all of the important details necessary for the general case. For full proofs with regard to the general definition of simulation, see Appendix~\ref{sec:generalsimulations}.

\begin{definition}[Special case of definition in~\cite{CMP18}; variant of definition in~\cite{BH17}]
	\label{dfn:specialsim}
	We say that $H'$ is a $(\Delta,\eta,\epsilon)$-simulation of $H$ if there exists a local isometry $V = \bigotimes_i V_i$ such that
	\begin{enumerate}
		\item
		There exists an isometry $\widetilde{V}$ such that $\widetilde{V} \widetilde{V}^\dag = P_{\le \Delta(H')}$, where $P_{\le \Delta(H')}$ is the projector onto the space of eigenvectors of $H'$ with eigenvalues less than $\Delta$, and $\|\widetilde{V} - V\| \le \eta$;
		\item
		$\| H'_{\le \Delta} - \widetilde{V}H\widetilde{V}^\dag \| \le \epsilon$, where $H'_{\le \Delta}=  P_{\le \Delta(H')} H'  P_{\le \Delta(H')}$.
	\end{enumerate}
	We say that a family $\mathcal{F}'$ of Hamiltonians can simulate a family $\mathcal{F}$ of Hamiltonians if, for any $H \in \mathcal{F}$ and any $\eta,\epsilon >0$, and $\Delta \ge \Delta_0$ for some $\Delta_0 > 0$, there exists $H' \in \mathcal{F}'$ such that $H'$ is a $(\Delta,\eta,\epsilon)$-simulation of $H$.
	We say that the simulation is efficient if, for $H$ acting on $n$ qudits, $\|H'\| = \poly(n,1/\eta,1/\epsilon,\Delta)$; $H'$ and $\set{V_i}$ are computable in polynomial-time given $H$, $\Delta$, $\eta$ and $\epsilon$ and provided that $\Delta, 1/\eta, 1/\epsilon$ are $O(\poly~n)$; and each isometry $V_i$ maps from at most one qudit to $O(1)$ qudits.
\end{definition}

\noindent We remark that unlike in~\cite{CMP18}, here we have the additional requirement that the local isometry $V$ is efficiently computable. This ensures that given some input Hamiltonian $H$ and local observable $A$, we can use the notion of simulation to efficiently produce a simulating Hamiltonian $H'$ \emph{and} a simulating observable $A'$ (see proof of Lemma~\ref{lem:simulations=reductions} below). As far as we are aware, all known constructions satisfying the notion of efficient simulation from~\cite{CMP18} fulfill this additional requirement (see proof of Theorem~\ref{thm:S-APXSIM} for examples).

Note that eigenvalues are preserved up to a small additive factor $\epsilon$ in a simulation, but that the YES instance in the definition of $\app$ is not robust to such perturbations of eigenvalues when the spectral gap is very small.
We therefore do not expect to show directly that hardness of $\app$ is preserved by simulations, and instead we work with the problem $\apptwo$.
(Recall though, an instance of $\apptwo$ trivially reduces to one of $\app$ with no modifications. Thus, if $\apptwo$ is hard for some family of Hamiltonians, then so too is $\app$.)
Let $\Fapptwo$ denote the problem $\apptwo$ restricted to Hamiltonians taken from the family $\mathcal{F}$.

\begin{lemma}[Simulations preserve hardness of $\apptwo$]
	\label{lem:simulations=reductions}
	Let $\mathcal{F}$ be a family of Hamiltonians which can be efficiently simulated by another family $\mathcal{F}'$.
	Then, $\Fapptwo$ reduces to $\Fapptwo[F']$ via polynomial-time many-one reductions.
\end{lemma}

Here, we provide a proof only for the special case where the simulation is of the form given in Definition~\ref{dfn:specialsim}; for a full proof of the general case, see Appendix~\ref{sec:generalsimulations}. 
\begin{proof}
	Let $\Pi = (H,A,k,\ell,a,b,\delta)$ be an instance of $\Fapptwo$. We will demonstrate that one can efficiently compute $H' \in \mathcal{F}'$ and $A',k',\ell',a',b',$ and $\delta'$ such that $\Pi' = (H',A',k',\ell',a',b',\delta')$ is a YES (respectively NO) instance of $\apptwo$ if $\Pi$ is a YES (resp. NO) instance of $\apptwo$; further, we will have that $\ell'\in O(\ell), a'=a+(b-a)/3,b'=b-(b-a)/3$ and $\delta-\delta'\geq 1/\poly(n)$. To do so, we shall pick parameters $\Delta, \eta, \epsilon$ so that $\Delta,1/\eta,1/\epsilon$ are $O(\poly~n)$, upon which the definition of efficient simulation (Definition~\ref{dfn:specialsim}) guarantees we can efficiently compute a Hamiltonian $H'$ being a $(\Delta, \eta, \epsilon)$-simulation of $H$, which we claim will preserve YES and NO instances $H$.

 Let us leave $\Delta, \eta, \epsilon$ arbitrary for now, and assume we have a simulation of the form given in Definition~\ref{dfn:specialsim}. Then, there exists an isometry $\Vt:\cH\rightarrow \cH'$ ($\cH$ and $\cH'$ are the spaces $H$ and $H'$ act on, respectively) which maps onto  the space of eigenvectors of $H'$ with eigenvalues less than $\Delta$, i.e. onto $S_{\le\Delta}:=\Span\{\ket{\psi}:H'\ket{\psi}=\lambda\ket{\psi}, \lambda \le\Delta\}$. In addition, $\Vt$ satisfies $\|\Vt-\bigotimes_iV_i\|\leqslant \eta$ and  $\|H_{\leqslant \Delta}-\Vt H\Vt^{\dagger}\|\leqslant \epsilon$.
		
	Let $\ket{\psi'}$ be a low-energy state of $H'$ satisfying $\bra{\psi'}H'\ket{\psi'} \le \lambda(H')+ \delta'$ for $\delta'$ to be set later. First, we show that $\ket{\psi'}$ is close to a state $\Vt \ket{\psi}$ where $\ket{\psi}$ is a low-energy state of $H$; then, we will show that there exists an observable $A'$, depending only on $A$ and the isometries $V_i$,  such that $\bra{\psi'}A' \ket{\psi'}$ approximates $\bra{\psi}A\ket{\psi}$ for any choice of $\ket{\psi}$. Since by Definition~\ref{dfn:specialsim} $A$ is efficiently computable, our choice of $A'$ will be as well.
	
	Let $\ket{\phi}=P_{\le \Delta(H')} \ket{\psi'}/\|P_{\le \Delta(H')} \ket{\psi'}\|$ be the (normalized) component of $\ket{\psi'}$ in $S_{\le \Delta}$.
	By Lemma~\ref{lem:lowenergydistance}, we have
	\[
	\frac{1}{2} \left\| \proj{\psi'}-\proj{\phi}\right\|_1 \le \sqrt{\frac{\delta'}{\Delta- \lambda(H')}} .
	\]
	
	Since $S_{\le \Delta}= \text{Im}(\widetilde{V})$, there must exist a state $\ket{\psi}$ in $\cH$ such that $\widetilde{V}\ket{\psi}=\ket{\phi}$; next, we will show that $\ket{\psi}$ has low-energy with respect to $H$.
	Note that $\ket{\psi'}= \sqrt{p}\ket{\phi}+\sqrt{1-p}\ket{\phi^{\perp}}$ for some $p \in [0,1]$ and a state $\ket{\phi^{\perp}}$ in $S_{\le \Delta}^{\perp}$ which has higher energy: $\bra{\phi^{\perp}} H' \ket{\phi^{\perp}}\ge \Delta \ge \bra{\phi} H' \ket{\phi}$. Therefore,
	\[
	\bra{\psi'}H'\ket{\psi'}=p\bra{\phi} H' \ket{\phi} +(1-p)\bra{\phi^{\perp}} H' \ket{\phi^{\perp}}\geqslant \bra{\phi} H' \ket{\phi} ,
	\]
	which implies that
	\begin{align}
	\bra{\psi}H\ket{\psi} -\bra{\psi'}H'\ket{\psi'} & \leqslant \bra{\psi}H\ket{\psi}- \bra{\phi}H'\ket{\phi} \\
	& =\bra{\phi}\Vt H\Vt^{\dagger}\ket{\phi}- \bra{\phi}H'\ket{\phi} \\
	& \leqslant \| H'_{\le \Delta}-\Vt H\Vt^{\dagger}\| \leqslant\epsilon .
	\end{align}

	So, $\bra{\psi}H\ket{\psi}\leqslant\lambda(H')+\delta'+\epsilon\leqslant\lambda(H)+\delta'+2\epsilon$, where the final inequality follows from Lemma 27 of \cite{CMP18}, which roughly states that eigenvalues are preserved up to error $\epsilon$ in a simulation (in particular, the minimum eigenvalues satisfy $\abs{\lambda(H')-\lambda(H)}\leq \epsilon$).
	
	For any local measurement $A_S$ acting on subset of $S$ qubits $\cH_S$ {(here $\cH_S$ is the Hilbert space for qudits in set $S\subseteq[n]$)}, we can define the local measurement $A'_S=V_SA_SV_S^{\dagger}$ on $\cH'_S$  where $V=\bigotimes V_i$ is the local isometry in the definition of simulation and $V_S:=\bigotimes_{i\in S} V_i$. Note that $A_S'$ acts only on the $O(\abs{S})$ qudits which $V_S$ maps to.
	Furthermore, $V^{\dagger}(A_S' \otimes I)V=A_S \otimes I$ and so
	\begin{align}
	|\bra{\psi'}A_S' \otimes I\ket{\psi'}-&\bra{\psi}A_S \otimes I \ket{\psi}|=|\bra{\psi'}A_S' \otimes I\ket{\psi'}-\bra{\psi}V^{\dagger}(A_S' \otimes I)V\ket{\psi}|\\
	&\leqslant \|A'_S\| \| \proj{\psi'}- V\proj{\psi}V^{\dagger}\|_1 \\
	&\leqslant \|A_S\| \left(\| \proj{\psi'}- \proj{\phi}\|_1 +\| \Vt \proj{\psi}\Vt^{\dagger}-V\proj{\psi}V^{\dagger}\|_1\right)\\
	& \leqslant \|A_S\|\left(\| \proj{\psi'}- \proj{\phi}\|_1+2\|\Vt -V\|\right) \label{eq:afterVtapprox}\\
	&\leqslant \|A_S\|\left(2\sqrt{\frac{\delta'}{\Delta-\lambda(H')}}+2\eta\right)
	\end{align}
	
	where to get to (\ref{eq:afterVtapprox}), we have used the triangle inequality to bound:
	\begin{align}
	\| \Vt \proj{\psi}\Vt^{\dagger}-V\proj{\psi}V^{\dagger}\|_1
	&\le
	\| \Vt \proj{\psi}\Vt^{\dagger}-V\proj{\psi}\Vt^{\dagger}\|_1+
	\| V \proj{\psi}\Vt^{\dagger}-V\proj{\psi}V^{\dagger}\|_1 \\
	&= \|\Vt-V\|\left(\|\proj{\psi}\Vt^\dagger\|_1+\|V\proj{\psi}\|_1\right)
	= 2\|\Vt-V\|
	\end{align}
	
	Therefore, to ensure that $\Pi'$ is a YES (resp. NO) instance if $\Pi$ is a YES (resp. NO) instance, we will choose $a'=a+(b-a)/3$ and $b'=b-(b-a)/3$. Choosing $\delta', \Delta, \epsilon, \eta$ such that
	\[
	0<\delta'+2\epsilon < \delta \qquad \text{ and } \qquad  0< \|A\|\left(2\sqrt{\frac{\delta'}{\Delta-\lambda(H')}}+2\eta\right) < \frac{b-a}{3}
	\]
	completes the proof.
\end{proof}

As a corollary of our results, we obtain Theorem~\ref{thm:S-APXSIM}, which gives a complete classification of the complexity of $\app$ when restricted to families of Hamiltonians and measurements built up from a set of interactions $\cS$. We restate it here for convenience:

\begin{reptheorem}{thm:S-APXSIM}
	Let $\mathcal{S}$ be an arbitrary fixed subset of Hermitian matrices on at most 2 qubits. Then the $\app$ problem, restricted to Hamiltonians $H$ and measurements $A$ given as a linear combination of terms from $\cS$ { and the identity $I$}, is
	\begin{enumerate}
		\item in $\class{P}$, if every matrix in $\mathcal{S}$ is 1-local;
		\item $\Plog{NP}$-complete, if $\cS$ does not satisfy the previous condition and there exists $U \in SU(2)$ such that $U$ diagonalizes all 1-qubit matrices in $\mathcal{S}$ and $U^{\otimes 2}$ diagonalizes all 2-qubit matrices in $\mathcal{S}$;
		\item $\Plog{StoqMA}$-complete, if $\cS$ does not satisfy the previous condition and there exists $U \in SU(2)$ such that, for each 2-qubit matrix $H_i \in \mathcal{S}$, $U^{\otimes 2} H_i (U^{\dag})^{\otimes 2} = \alpha_i Z^{\otimes 2} + A_iI + IB_i$, where $\alpha_i \in \R$ and $A_i$, $B_i$ are arbitrary single-qubit Hermitian matrices;
		\item $\Plog{QMA}$-complete, otherwise.
	\end{enumerate}
\end{reptheorem}

\begin{proof} We first discuss containment in the claimed complexity classes, and then hardness.

\paragraph{Containment.} In the first case it is trivial to simulate the outcome of 1-local measurements on the ground state of a 1-local Hamiltonian, as the ground state is an easily calculated product state. For the other three cases, it was shown in \cite{CM16} and \cite{BH17}, that $k$-LH for these three families of Hamiltonians is complete for the classes $\class{NP}, \class{StoqMA}, \class{QMA}$, respectively. Therefore, by Lemma~\ref{lem:APXinPXlog}, $\app$ is contained in $\Plog{NP}, \Plog{StoqMA}$ and $\Plog{QMA}$, respectively. (Note that the precondition of Lemma~\ref{lem:APXinPXlog} is met, i.e. for $H$ and $A$ given as a linear combination of terms from $\cS$ and $I$, we have that \klh for $\alpha H+ \beta A$ is contained in the respective complexity class of $\class{NP}$, $\class{StoqMA}$, or $\class{QMA}$, for any $0\le \alpha, \beta \le \poly(n)$, and for all $k\geq 1$.)

\paragraph{Hardness.} Starting with the referenced completeness results of~\cite{CM16,BH17} above, we now wish to show $\app$ is hard for $\Plog{NP}, \Plog{StoqMA}$ and $\Plog{QMA}$ for cases 2--4 of our claim. At first glance, it may seem that Theorem~\ref{thm:APXcompleteness} already yields this result, since that theorem says that $\app$ is $\Plog{C}$-complete when restricted to $k$-local Hamiltonians and observables from a family $\cF$. Unfortunately, however, a precondition of Theorem~\ref{thm:APXcompleteness} is that $\cF$ must contain all classical (i.e. diagonal in standard basis) Hamiltonians, which is \emph{not} necessarily true for cases 2--4 of our claim here. Thus, some work is required get the hardness claims of cases 2--4 here.

To achieve this, we first apply Lemma~\ref{lem:APXisPparXhard} to conclude that \apptwo\ is hard for classes $\Ppar{NP}, \Ppar{StoqMA}$ and $\Ppar{QMA}$ for the families of classical, stoquastic and arbitrary local Hamiltonians, respectively. (In contrast to the Hamiltonians of cases 2--4 of our claim here, the sets of classical, stoquastic and arbitrary local Hamiltonians \emph{do} contain all diagonal Hamiltonians, and thus satisfy the preconditions of Lemma~\ref{lem:APXisPparXhard}.) We then use simulations, in combination with Lemma~\ref{lem:simulations=reductions}, to reduce the sets of classical, stoquastic, and arbitrary local Hamiltonians to the Hamiltonians in cases 2,3,4 of our claim here, respectively.

Specifically, it was shown in \cite{CMP18} that the three families of Hamiltonians in cases 2--4 of our claim can efficiently simulate all classical, stoquastic and arbitrary local Hamiltonians, respectively, via some local isometry $V$ (see Definition~\ref{dfn:specialsim}). It follows by Lemma~\ref{lem:simulations=reductions} (which states that simulations act like hardness reductions) that $\apptwo$ is hard for $\Ppar{NP}, \Ppar{StoqMA}$ and $\Ppar{QMA}$ respectively, with respect to (using the notation of Lemma~\ref{lem:simulations=reductions}) a local observable $A'$ (in the larger, simulating, space) such that $A'=VAV^{\dagger}$ (where in our case $A$ will equal Pauli $Z$ due to the proof of Lemma~\ref{lem:APXisPparXhard}). The only obstacle to achieving our current claim is that we also require $A'$ to be chosen as a linear combination of terms from $\cS$ and $I$. This is what the remainder of the proof shall show.

\emph{Observation (*).} To begin, note the proof of Lemma~\ref{lem:APXisPparXhard} used single qubit observable $Z$, since we encoded the P machine's output in a single bit, which we assumed was set to $\ket{0}$ for ``reject'' and $\ket{1}$ for ``accept''. However, without loss of generality, we may alter the starting P machine to encode its output in some more general function on two bits, such as the parity function. (For example, the P machine can be assumed to output a $2$-bit string $q$, such that $q$ has odd parity if and only if the P machine wishes to accept.) We use this observation as follows. Consider any classical observable $A$ with two distinct eigenvalues $\lambda_x<\lambda_y$ corresponding to eigenstates $\ket{x}$ and $\ket{y}$, respectively, for distinct strings $x,y\in\set{0,1}^2$. Then, assuming the specification of $A$ is independent of the number of qubits in the system (thus, $A$ is specified to within constant bits of precision, and so $\lambda_y-\lambda_x\in\Theta(1)$), if we set the P machine to output $x$ when it wishes to accept and $y$ when it wishes to reject, a measurement with observable $A$ suffices to distinguish these two cases. With this observation in hand, we consider cases 2--4 of our claim, in particular with respect to the action of isometry $V$.

\emph{Case 2: $\Ppar{NP}$-completeness.} First note that in this case we can assume without loss of generality that all interactions in $\cS$ are diagonal (by performing  a global basis change of $U^{\otimes n}$ if necessary) . Since we are not in the first case we know also that there is a 2-local interaction in $\cS$ with at least two distinct eigenvalues. By Observation (*), it will suffice to simulate such an observable on a particular pair of qubits in the original system; call this operator $A$. For the  $\Plog{NP}$ case, the isometry $V$ appends  some ancilla qubits in a computational basis state (in the $U^{\otimes n}$ basis) \cite{DlCC16}. We can therefore choose $A'$ to be the same 2-local observable $A$, but acting on the corresponding qubits in the larger, simulating system; that is, if we let $A'=A \otimes I$ (where the identity term acts on the ancilla qubits), then $V^{\dagger} A' V=A$ as desired.

\emph{Case 3: $\Ppar{StoqMA}$-completeness.} For the third case, one can check that the reductions in \cite{BH17} correspond to a simulation with an isometry $V$ which maps each qubit  $\ket{0} \mapsto \ket{0011}$ and $\ket{1} \mapsto \ket{1100}$ and appends some additional ancilla qubits in a computational basis state (see discussion in Section 9.4 of \cite{CMP18}). Thus, a classical $2$-local observable $Z\otimes Z+\diag(A)\otimes I+I\otimes \diag(B)$ (which we may use by Observation (*)) can be simulated in the larger, simulating space on physical qubits $1,2,3,4 $ (logical qubit $1$) and $5,6,7,8$ (logical qubit $2$) via:
\[
    V^{\dagger}(Z_1Z_5+A_1+B_5)V= Z\otimes Z+\diag(A)\otimes I+I\otimes \diag(B),
\]
where $\diag(A)$ denotes the diagonal part of $A$, i.e. $\diag(A)=\sum_{i=0}^1 \proj{i}A\proj{i}$. Thus, measuring observable $(Z_1Z_5+A_1+B_5)$ on the larger, simulating Hamiltonian $H'$ (which has the desired form of Case 3 here) is equivalent to measuring $Z\otimes Z+\diag(A)\otimes I+I\otimes \diag(B)$ on the starting Hamiltonian $H$ in the simulation (again, using notation of Lemma~\ref{lem:simulations=reductions}).
	
\emph{Case 4: $\Ppar{QMA}$-completeness.}	The final case is slightly more complicated. When showing that these Hamiltonians are universal, the one step with a non-trivial isometry is simulating $\{X,Z,XX,ZZ\}$-Hamiltonians with $\{XX+YY\}$-Hamiltonians or $\{XX+YY+ZZ\}$-Hamiltonians in Theorem 41 of \cite{CMP18}.
	In both of these cases, the isometry $V$ maps each qubit via action
	\[\ket{0} \mapsto \ket{\Psi^-}_{13}\ket{\Psi^-}_{24} \qquad
	\ket{1} \mapsto \tfrac{2}{\sqrt{3}}\ket{\Psi^-}_{12}\ket{\Psi^-}_{34}
	- \tfrac{1}{\sqrt{3}}\ket{\Psi^-}_{13}\ket{\Psi^-}_{24}.\]
	In the proof of Theorem 41 of \cite{CMP18}, it is shown that a single $Z$ observable can be reproduced by choosing $A=h_{13}$ (where either $h=XX+YY$ or $h=XX+YY+ZZ$), that is $V^{\dagger}h_{13}\otimes I_{24}V$ is proportional to $Z$.
	
	The proof is completed by Corollary~\ref{cor:complexityequalities} (i.e. logarithmic adaptive queries are equivalent to polynomially many parallel queries).
\end{proof}
%=========================================================================================
%  SECTION: SPATIALLY SPARSE
%=========================================================================================

\section{Spatially sparse construction}
\label{sec:spatiallysparse}

We now combine the tools developed in the previous sections to study the complexity of \app\ for physical Hamiltonians. Our approach is to show that $\apptwo$ is $\PQMApar$-hard even for Hamiltonians on a \emph{spatially sparse interaction graph}, defined below:

\begin{definition}[Spatial sparsity~\cite{OT08}]
  A spatially sparse interaction (hyper)graph $G$ on $n$ vertices is defined as a (hyper)graph in which \begin{inparaenum}
  \item every vertex participates in $O(1)$ hyper-edges;
  \item there is a straight-line drawing in the plane such that every hyper-edge overlaps with $O(1)$ other hyper-edges and the surface covered by every hyper-edge is $O(1)$.
  \end{inparaenum}
\end{definition}

\begin{lemma}
\label{lem:spatiallysparse}
$\apptwo$ is $\PQMApar$-hard even when $b-a=\Omega(1)$, the observable $A$ is 1-local (single-qubit), and the Hamiltonian $H$ is 4-local and is restricted to a spatially sparse interaction graph.
\end{lemma}

Here, we adapt the proof of Lemma~\ref{lem:APXisPparXhard}.
Recall that the Hamiltonian $H$ in Lemma~\ref{lem:APXisPparXhard} is composed of two parts $H=H_1+H_2$, where $H_2$ uses (a simplification of) Ambainis's query Hamiltonian on each of the registers $\X_i \otimes \Y_i$ to encode the answer to that query into the state of $\X_i$ (see Equation~(\ref{eqn:Hdef})), and $H_1$ encodes the evolution of the $\Poly$ circuit using the Cook-Levin construction on the $\W$ register (controlling on the states of the $\X_i$ registers). This is represented by Figure~\ref{fig:prevstructure}.

We arrange the qubits of the $\W$ register on a square lattice and note that $H_1$ is already manifestly spatially sparse.
This is one of the advantages of using the Cook-Levin construction over the Kitaev history state construction.
Furthermore, the Hamiltonian $H_{\Y_i}$, corresponding to the $i$-th $\QMA$ query, can be chosen to be spatially sparse -- in fact it can be chosen to have its interactions on the edges of a 2D square lattice~\cite{OT08}, and so we also lay out the qubits of each $\Y_i$ register on a square lattice.

But the interaction graph of this Hamiltonian is still far from spatially sparse because in (the modified version of) Ambainis's query Hamiltonian $H_2$, every qubit of $\Y_i$ interacts with $\X_i$.
We will solve this problem by replacing each single qubit $\X_i$ register with a multi-qubit register of $n_i$ qubits labeled by $\{\X_i(j)\}_{j=1}^{n_i}$, for $n_i$ the number of qubits of $\Y_i$. We spread out the qubits of the $\X_i$ register in space around the $\Y_i$ register, and modify $H_2$ so that each term is controlled only on a nearby qubit in the $\X_i$ register.
To make this work we need to introduce a third term $H_3$ which ensures that all the qubits in each $\X_i$ register are either all $\ket{0}$ or all $\ket{1}$.

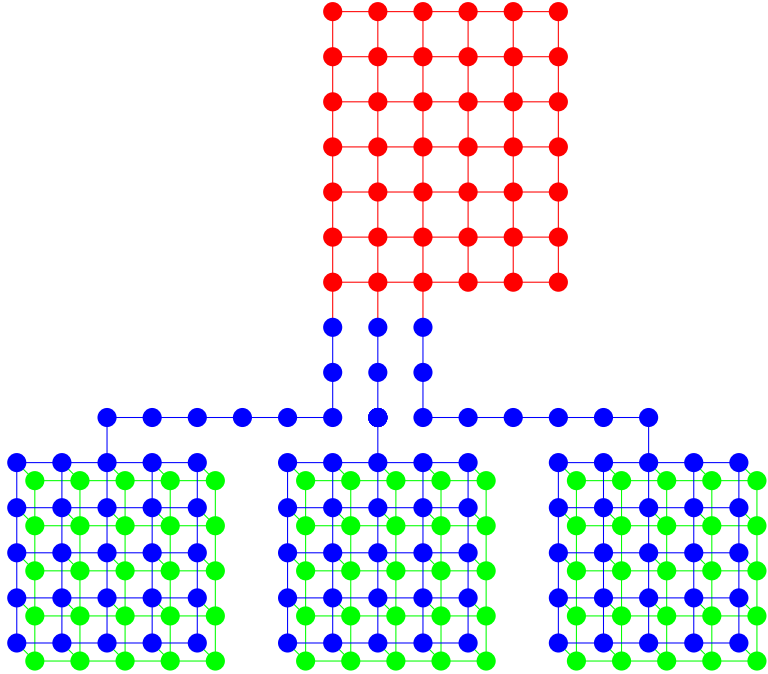
\begin{figure}[h!]
\centering
\begin{tikzpicture}[scale=0.6]

\tikzstyle{H3}=[circle,scale=0.7,fill=blue];
\tikzstyle{H2}=[circle,scale=0.7,fill=green];
\tikzstyle{H1}=[circle,scale=0.7,fill=red];
\foreach \a in {-1,0,1}
{\begin{scope}[shift={(\a*6,0)}]
	\begin{scope}[shift={(0.4,-0.4)}]
		\draw[green] (1,1) grid (5,5);
		\foreach \x in {1,...,5}{
			\foreach \y in {1,...,5}{
				\node[H2] at (\x,\y){};
				\draw[green] (\x-0.4,\y+0.4) -- (\x,\y);}}
	\end{scope}
	\foreach \x in {1,...,5}{
		\foreach \y in {1,...,5}{
			\node[H3] at (\x,\y) {};
		}
		\node[H3] at (3-\a*\x,6){};
	}
	\node[H3] at (3,6){};
	\draw[blue] (3,5) -- (3,6)--(3-5*\a,6)--(3-5*\a,8);
	\draw[blue] (1,1) grid (5,5);
	\draw[red] (3-5*\a,8) -- (3-5*\a,9);
	\foreach \d in {7,...,8}{
		\node[H3] at (3-5*\a,\d){};
	}
\end{scope}
}
\begin{scope}[shift={(1,8)}]
	\draw[red] (1,1) grid (6,7);
	\foreach \x in {1,...,6}{
		\foreach \y in {1,...,7}{
			\node[H1] at (\x,\y) {};
}}
\end{scope}
\end{tikzpicture}
\caption[Geometric structure of Hamiltonian used in Lemma~\ref{lem:spatiallysparse}]{(Color figure) Geometric structure of total Hamiltonian $H= {\color{red} H_1}+{\color{green} H_2}+{\color{blue}H_3}$ for the case of $3$ queries. In words, $H_1$ is the top square, $H_3$ is the set of connecting wires, along with the bottom three squares to which they are connected. $H_2$ is the remaining set of three squares at the bottom of the diagram.}
\label{fig:2Dstructure}
\end{figure}

\begin{proof}[Proof of Lemma~\ref{lem:spatiallysparse}.]
We will construct a Hamiltonian on the registers $\W$, $\X_i$ and $\Y_i$ for $i\in \{1,\dots m\}$, for which the problem $\apptwo$ encodes the output of a $\PQMApar$ circuit, where $m$ is the number of parallel queries to the QMA oracle.

Let the qubits of $\W$ and $\Y_i$ be arranged on distinct parts of a square lattice. For each qubit of $\Y_i$, there is a corresponding qubit in $\X_i$, and $\X_i$ contains a path of qubits leading from $\Y_i$ to $\W$. See Figure~\ref{fig:2Dstructure} for an example layout in the case $m=3$.

Let $E_i$ be the set of edges of the square lattice of qubits of $\Y_i$ (i.e. not including the edges connecting $\Y_i$ to $\X_i$ in Figure~\ref{fig:2Dstructure}) and let $H_{\Y_i}=\sum_{(j,k) \in E_i} h^{i}_{\Y_i(j,k)}$ be a 2D nearest neighbor Hamiltonian on $\Y_i$ corresponding to the $i$-th query.
We have used the subscript notation $\Y_i(j,k)$ to denote the action of an operator on the $j$-th and $k$-th qubits of the $\Y_i$ register.
$H_{\Y_i}$ has ground state energy less than $a_i$ if query $i$ is a YES instance and energy greater than $b_i$ in a NO instance. Then, let $H_2=\sum_i H_2^{(i)}$ where

\[
	H_2^{(i)}=\frac{a_i+b_i}{2} \proj{0}_{\X_i(1)} \otimes I_{\mathcal{Y}_i} + \sum_{(j,k) \in E_i} \left( \proj{1}_{\X_i(g(j,k))} \otimes h^{i}_{\Y_i(j,k)}\right) ,
\]
where $g(j,k)$ is the location of the ``nearest'' qubit in $\X_i$ to edge $(j,k)$ in $\Y_i$. Here, the choice ``nearest'' is somewhat arbitrary; for concreteness, one can set $g(j,k)=j$, i.e. pick the vertex in $\spa{X}_i$ which aligns with the first coordinate of the edge $(j,k)$. (In this sense, Figure~\ref{fig:2Dstructure} is not entirely accurate, since it depicts the $3$-local constraint $\proj{1}_{\X_i(g(j,k))} \otimes h^{i}_{\Y_i(j,k)}$ as a pair of $2$-local constraints. This is done solely for the purpose of simplifying the illustration, as otherwise one would need to draw hyperedges of size $3$.)

Let $H_1=H_{\text{prop}}+H_{\text{in}}$ be the Cook-Levin Hamiltonian where $H_{\text{prop}}$ is exactly as in Lemma~\ref{lem:APXisPparXhard}.
Let $H_{\text{in}}$ initialize the qubits of the first ($t=1$) row of the qubits in $\W$.
For each query $i$, we have a penalty term $\proj{1}_{\X_i(1)}\proj{0}+\proj{0}_{\X_i(1)}\proj{1}$ which effectively copies the state of $\X_i(1)$, the qubit in $\X_i$ nearest to $\W$, onto the $i$-th qubit of the first row of $\W$. For all the remaining qubits in the first ($t=1$) row of $\W$, we have a penalty term $\proj{1}$, effectively initializing the qubit into the $\ket{0}$ state.

Restricted to the subspace $\cH$ where each $\X_i$ register is either all $\ket{0}$ or all $\ket{1}$, $H_1+H_2$ is exactly the same Hamiltonian as in Lemma~\ref{lem:APXisPparXhard}.
It remains to give a high energy penalty to all other states not in this subspace. We do this with $H_3=\sum_{i=1}^m  H_3^{(i)}$ where each term $H_3^{(i)}$ acts on $\X_i$:
\[H_3^{(i)}=\Delta_i\sum_{(j,k) \in G_i} \left(\proj{0}_{\X_i(j)} \proj{1}_{\X_i(k)}+ \proj{1}_{\X_i(j)} \proj{0}_{\X_i(k)}\right)\]
where $G_i$ is the set of edges between the qubits of the $\X_i$ register.
$G_i$ consists of edges between nearest neighbors on the square lattice $E_i$ and on the path of qubits from $\Y_i$ to $\W$.
The overall Hamiltonian $H=H_1+H_2+H_3$ is therefore spatially sparse.

$H_3^{(i)}$ is a classical Hamiltonian, so all of its eigenstates can be taken to be of form $\ket{x}$ for some $x \in \{0,1\}^{n_i}$.
Its ground space $\spa{G}_i$ contains $\ket{0}^{\otimes n_i}$ and $\ket{1}^{\otimes n_i}$; and all states in $\spa{G}_i^\perp$ have energy at least $\Delta_i$. Choosing $\Delta_i > \delta + \sum_{(j,k) \in E_i} \|h^i_{\Y_i(j,k)}\|$ ensures that all states in $\spa{G}_i^{\perp}$ have energy greater than $\lambda(H)+\delta$.

Then $H=H_1+H_2+H_3$ is block diagonal with respect to the split of each subspace $\spa{G}_i\oplus\spa{G}_i^\perp$; restricted to the spaces $\spa{G}_i$, $H$ is exactly the Hamiltonian from Lemma~\ref{lem:APXisPparXhard}, and all states in spaces $\spa{G}_i^\perp$ have energy greater than $\lambda(H)+\delta$.
The result then follows just as in the proof of Lemma~\ref{lem:APXisPparXhard}.
\end{proof}

Finally we restate Theorem~\ref{thm:simofsparsehard} which shows $\app$ is hard not only for families of Hamiltonians which are universal -- that is, families that can efficiently simulate any $k$-local Hamiltonian -- but also for more restricted families of Hamiltonians which can only efficiently simulate the family of spatially sparse Hamiltonians. As stated in Section~\ref{sscn:results2}, this then yields the desired hardness results for \app\ on physical Hamiltonians such as the Heisenberg interaction on a 2D lattice (see, e.g., Corollary~\ref{cor:1}).

\begin{reptheorem}{thm:simofsparsehard}
	Let $\mathcal{F}$ be a family of Hamiltonians which can efficiently simulate any spatially sparse Hamiltonian. Then, $\app$ is $\PQMA$-complete even when restricted to a single-qubit observable and a Hamiltonian from the family $\mathcal{F}$.
\end{reptheorem}

\begin{proof}
This follows from Lemma~\ref{lem:simulations=reductions}, Lemma~\ref{lem:spatiallysparse} and Corollary~\ref{cor:complexityequalities}.
\end{proof}

%%%%%%%%%%%%%%%%%%%%%%%%%%%%%%%%%%%%%%%
%%%%%%%%%%%%%%%%%%%%%%%%%%%%%%%%%%%%%%%
%=========================================================================================
%  SECTION: Simulating measurements on a 1D line
%=========================================================================================

\section{Simulating measurements on a 1D line}\label{sec:1D}

In this section, we show that \app\ remains $\PQMA$-complete even on a line. Below, we reproduce the statement of the main theorem of this section for convenience.

\begin{reptheorem}{thm:apxsim-1dline-complete}
	$\app$ is $\PQMA$-complete even when restricted to Hamiltonians on a 1D line of 8-dimensional qudits and single-qudit observables.
\end{reptheorem}

We prove Theorem~\ref{thm:apxsim-1dline-complete} in three sections. We first describe our construction in Section~\ref{ssec:1d-construction}. We then show correctness of the construction in Section~\ref{ssec:1D-correctness}, with the proofs of various lemmas deferred to Section~\ref{sscn:lemmas}.

\subsection{Our 1D hardness construction}\label{ssec:1d-construction}

We give a reduction from $\Ppar{QMA}$ to $\apptwo$, which by Theorem~\ref{thm:APXcompleteness} and the fact that $\apptwo$ trivially reduces to $\app$ yields $\PQMA$-hardness of \app. Let $\Pi$ be a $\Ppar{QMA}$ computation which takes in an input of size $n$ and which consists of a uniformly generated polynomial-size classical circuit $C$ making $m = O(\log n)$ $\klh[2]$ queries $\pi_i := (H_i, a_i, b_i)$ to a $\QMA$ oracle. As in Section~\ref{ssec:hardness}, we treat the ``answer register'' in which $C$ receives answers to its $m$ queries as a proof register.

Our high-level approach consists of three steps: (1) construct a ``master'' circuit $V$ composed of the verification circuits $V_i$ corresponding to each query $\pi_i$ and of the circuit $C$; (2) run $V$ through the 1D circuit-to-Hamiltonian construction of~\cite{HNN13} to obtain a 1D Hamiltonian $G$ with local dimension $8$ constructed such that the low-energy space $\spa{S}$ of $G$ must consist of history states (of the form described in~\cite{HNN13}); and (3) carefully add additional $1$-local penalty terms acting on the output qubits corresponding to each verification circuit $V_i$ to obtain final Hamiltonian $H$ such that the low-energy space must encode satisfying proofs to each $V_i$ whenever possible. This final step of ``fine-grained splitting'' of $S$ forces the output qubits of the circuits $V_i$ to encode correct answers to query $\pi_i$, and thus the final circuit $C$ receives a correct proof, hence leading the history states of step (2) to encode a correct simulation of $\Pi$. The answer to the computation $\Pi$ can then be read off the ground state of $H$ via an appropriate single qudit measurement.

\paragraph{1. Construction of $V$.} Suppose each query $\pi_i$ has corresponding $\QMA$ verification circuit $V_i$.
Without loss of generality, we may henceforth assume that the completeness/soundness error of $V_i$ is at most $p \leq 2^{-n}$, for $p$ to be set later, by standard error reduction~\cite{AN02,MW05}; thus, if a particular query $(H_i,a_i,b_i)$ is valid (i.e. $\lambda(H)\notin(a_i,b_i)$), then either there exists a proof such that $V_i$ outputs YES with probability at least $1-p$ or no proof causes $V_i$ to output YES with probability greater than $p$.
Next, since $\Pi$ is a $\Ppar{QMA}$ computation, all queries and corresponding $V_i$ can be precomputed in polynomial-time. We view the ``master circuit'' $V$ as consisting of two phases:
\begin{enumerate}
	\item (Verification phase) Given supposed proofs for each query, $V$ runs all verification circuits $V_i$ in parallel, where $V_i$ acts on space $\spa{Y}_i\otimes\spa{W}_i\otimes\spa{X}_i$, for proof register $\spa{Y}_i$, ancilla register $\spa{W}_i$, and single-qubit output register $\spa{X}_i$.

	\item (Simulated classical phase) The simulated $\Poly$ circuit $C$ now receives the query answers $\spa{X}:=\spa{X}_1\otimes\cdots\otimes\spa{X}_m$ as its proof register as well as an ancilla register $\spa{W}_0$. It outputs a single qubit to an output register $\spa{X}_0$.
\end{enumerate}
This completes the construction of $V$, which acts on $\spa{Y}\otimes\spa{W}\otimes\spa{X}$, where $\spa{Y}=\bigotimes_{i=1}\spa{Y}_i$,$\spa{W}=\bigotimes_{i=1}\spa{W}_i$, and $\spa{X}=\bigotimes_{i=1}\spa{X}_i$.
Crucially, note that given a set of proofs in register $\spa{Y}$, $V$ does \emph{not} necessarily yield the same answer as $\Pi$, since a malicious prover could intentionally send a ``bad'' proof to a YES query, flipping the final answer of $V$.

\paragraph{2. Construction of $G$.} We now plug $V$ into the circuit-to-Hamiltonian construction of Hallgren, Nagaj, and Narayanaswami~\cite{HNN13} to obtain a nearest-neighbor 1D Hamiltonian $G' = \din\hin+\dprop\hprop+\dpen\hpen+\hout$, where $\din,\dprop$, and $\dpen$ are at most polynomials in $n$ which we will set as needed; we review this construction more closely below. Set $G=G'-\hout$, since in our setting the task of ``checking the output'' will be delegated to the observable $A$.
Note that as an intermediate step, \cite{HNN13} maps $V$ to a circuit $V'$ which it then maps to $G'$; we describe the role of $V'$ in the following review.
Our construction will make two trivial assumptions about the behavior of $V'$, including how it arranges its query answers between the verification phase and the simulated classical phase and how it stores its output in the final timestep; we defer details about these assumptions until we define our ``fine-grained splitting'' in step 3 and when we define our observable.

\paragraph{Review of 1D $\QMA$ construction~\cite{HNN13}.}
Suppose an arbitrary circuit $U$ acts on $n$ qubits. Begin by arbitrarily arranging these qubits along a line. The circuit $U$ is then ``linearized'', meaning it is mapped to a new circuit $U'$ which consists of $R$ rounds in which each round applies a sequence of $n-1$ two-qubit gates acting on nearest neighbors. The $i$-th gate in a round acts on qubits $(i,i+1)$. This ``linearization'' is achieved in polynomial time by inserting swap and identity gates as needed, and $U'$ is at most polynomially larger than $U$.

To reduce $U'$ to an instance of $\klh$, we wish to design a mapping similar to Kitaev's circuit-to-Hamiltonian construction for showing $\QMA$-hardness of $\klh[5]$ on general geometry \cite{KSV02}. In both settings, the goal is to design an $H$ which enforces a structure on any state in its low-energy space. In the construction of \cite{KSV02}, $H = \hin+\hprop+\hstab+\hout$, and the minimizing state of $H$ has the form of a \emph{history state}:
\[
	\ket{\eta} = \frac{1}{\sqrt{L+1}}\sum_{t=0}^LU_t\cdots U_1\ket{\psi}_{\spa{Y}}\ket{0\cdots 0}_W\ket{t}_C .
\]
Intuitively, $\hstab$ forces a structure on the clock register $C$ of basis states $\ket{0},\ket{1},\dots$, such that each will correspond to a timestep of $U$. Then, $\hin$ ensures the ancilla register $W$ is set to the all $\ket{0}$ state when $\ket{t}=\ket{0}$. The term $\hprop$ ensures that the workspaces entangled with timesteps $\ket{t}$ and $\ket{t+1}$ are related by the 2-qubit gate $U_{t+1}$. Together, these terms ensure that a minimizing state $\ket{\psihist}$ encodes a correct simulation of the circuit $U$, and that all low-energy states are close to $\ket{\psihist}$. In fact, a valid $\ket{\psihist}$ lies in the nullspace of $\hin+\hprop+\hstab$. Finally, $\hout$ penalizes the low-energy space if the output qubit has overlap with $\ket{0}$.

Now in the 1D setting, the goal remains the same: design $H$ such that the structure of its low-energy state is a superposition over a sequence of states corresponding to timesteps in the computation of $U'$. But, we now appear unable to entangle the workspace with a separate clock register using nearest neighbor interactions. Instead, the constructions of \cite{AGIK09,HNN13} employ qudits of higher dimension as a means to label the qubits, with each labeling encoding a particular timestep. \cite{HNN13} then doubles the number of qudits in order to lower the necessary number of labels. The construction of \cite{HNN13} thus maps $U'$ to a Hamiltonian $H=\hin+\hprop+\hout+\hpen$ acting on $2nR$ qudits of dimension $8$, where the qudits are arranged on a 1D line in $R$ blocks of $2n$ qudits (i.e. one block per round in $U'$).

Let us further describe the idea of labeling, or ``marking'', of qudits. For example, a qubit $\alpha\ket{0}+\beta\ket{1}$ may be encoded as $\alpha\ket{A}+\beta\ket{B}$ if that qubit is ready for a gate to be applied or as $\alpha\ket{C}+\beta\ket{D}$ if that round's gate has already been applied, where $\ket{A},\ket{B},\ket{C},\ket{D}$ are some basis states.
The possible configurations, or arrangements, of labels along the line form a set of orthogonal spaces. \cite{HNN13} thus introduces a Hamiltonian term $\hpen$ which enforces a set of ``legal configurations'' of the workspace, penalizing all other configurations. We then map each of the configurations which remain in the low-energy space of $H$ to timesteps in the computation of $U'$, effectively assigning the job of encoding the workspace in a particular timestep to a particular configuration of qudits. We note that the crucial feature of the set of legal configurations developed by \cite{HNN13} is that they are sufficiently identifiable solely by 2-local nearest neighbor checks\footnote{For clarity, in~\cite{HNN13} not all illegal configurations are immediately detectable by $\hpen$. Any such undetectable illegal configurations are instead shown to eventually evolve under $\hprop$ into detectable illegal configurations.} such that penalties can be correctly assigned when constructing 1D analogs of the terms $\hin,\hprop,\hout$. Similar to the general geometry case of \cite{KSV02}, the construction of \cite{HNN13} enforces that the nullspace of $\hin+\hprop+\hpen$ consists of history states
\begin{equation}\label{eqn:hist-state-1D}
	\ket{\psihist}=\frac{1}{\sqrt{L+1}} \sum_{t=0}^{L} \ket{\psi_t},
\end{equation}
such that $\ket{\psihist}$ is a superposition over states in each legal configuration, $\ket{\psi_0}$ encodes a properly initialized workspace, and each pair $\ket{\psi_t}$ and $\ket{\psi_{t+1}}$ are related according to the corresponding timestep of $U'$. Finally, again similar to the general geometry case, all low-energy states must be close to $\ket{\psihist}$ (we make these two claims explicit and give proofs in Lemma~\ref{lem:props}).

The full description of the labeling, the legal configurations, and their mapping to timesteps by \cite{HNN13} is rather involved. Here, we introduce sufficient details for our later analysis.
We begin with a single block of $2n$ qudits, where recall each block is used to encode a single round (taken from~\cite{HNN13}):
\begin{align}
	\bdry \gate \insi \parity \qubit \insi \parity \cdots \parity \qubit \insi \parity \qubit \blnk
	\bdry\label{eqn:ex1}
\end{align}
Recall the design of $U'$ began by arranging the qubits of $U$ arbitrarily on the line; the $i$-th qubit on that line corresponds to qudits $2i-1$ and $2i$ in \eqref{eqn:ex1}. Thus, each qubit of $U'$, henceforth denoted a \emph{logical qubit}, is encoded into two consecutive qudits. Each pair of qudits representing a logical qubit is depicted as separated by a $\parity$ for clarity.
The standard basis for each $8$-dimensional qudit is labeled by
\[
	\set{\ket{\blnk},\ket{\lmove},\ket{\insi},\ket{\dead},\ket{\gate_0},\ket{\gate_1},\ket{\qubit_0},\ket{\qubit_1}} ,
\]
where, as described earlier, the current state of a qudit can be used to encode a logical qubit and to label the qudit. The first four states should be thought of as 1-dimensional labels; they are used to ensure the correct propagation of the circuit and do not encode a logical qubit. The final four states are used to either label a qudit with $\gate$, in which case a logical qubit is encoded as a superposition of $\ket{\gate_0}$ and $\ket{\gate_1}$, or with $\qubit$, in which case a logical qubit is encoded as a superposition of $\ket{\qubit_0}$ and $\ket{\qubit_1}$. To make this example more concrete, a product state of $(\alpha\ket{0}+\beta\ket{1})^{\otimes n}$ on $n$ logical qubits could be encoded as
\begin{align}
	(\alpha\ket{\gate_0}+\beta\ket{\gate_1})\otimes\ket{\insi}\otimes (\alpha\ket{\qubit_0}+\beta\ket{\qubit_1})\otimes\ket{\insi}\otimes\cdots\otimes(\alpha\ket{\qubit_0}+\beta\ket{\qubit_1})\otimes\ket{\blnk}.\label{eqn:ex2}
\end{align}

Next, here is an example depicting multiple blocks (from Table 2 of~\cite{HNN13}):
\begin{align}
	\cdots \dead \dead \bdry \gate \insi \parity \qubit \insi \parity \qubit \blnk \bdry \blnk \blnk \parity \blnk \blnk \parity \blnk \blnk \bdry \blnk \blnk \cdots ,
\end{align}\label{eqn:ex3}
where the blocks are delineated by $\bdry$. The labels $\dead$ to the left depict ``dead'' qudits, while the labels $\blnk$ to the right depict ``unborn'' qudits. By construction, all logical qubits are encoded in a block between the dead and unborn labels. In this example, the logical qubits line up with the beginning of a new block, beginning with $\bdry\gate$ and ending with the first $\blnk\bdry$.

At a high level, the set of legal configurations is mapped to a sequence of timesteps as follows. The first timestep corresponds to a configuration similar to \eqref{eqn:ex1}, with $n$ logical qubits encoded in the leftmost block of $2n$ qudits, with no $\dead$ labels anywhere, and with the ``gate'' label $\gate$ on the first qudit. The second configuration has the $\gate$ label shifted to the right, on the second qudit. Next, the third configuration has the second qudit labeled $\qubit$ and the third qudit labeled $\gate$. This propagation of the $\gate$ label rightwards continues, with each step corresponding to another legal configuration, until it reaches the end of the block. As the $\gate$ passes between logical qubits $(i,i+1)$, the corresponding configurations map to timesteps $i$ and $i+1$ of round 1, and $\hprop$ enforces that configurations are related by the application of gate $U'_i$. Thus, when we reach a configuration with $\gate$ at the end of the block, i.e. $\gate\bdry$, all gates in the current round will have been applied.
Next, before encoding the next round of gates, our goal becomes to shift all of the logical qubits encoded in the current block rightwards $2n$ spots into the second block. To do this, the $\gate$ label becomes a special $\lmove$ label and moves to the left one spot at a time until it reaches the end of the logical qubits (here, the leftwards $\bdry$). As the label $\lmove$ moves left, it shifts each logical qubit to the right one spot, i.e. $\ket{\qubit\lmove}\rightarrow\ket{\lmove\qubit}$.
This process repeats, with a label propagating rightwards to the end of the logical qubits (now past the rightwards $\bdry$), then the label $\lmove$ propagating to the left, shifting logical qubits to the right, and so on, until the logical qubits have shifted entirely into the second block. Then, the gate label $\gate$ once again transitions down the line, with successive configurations encoding the second round of gates of $U'$. Throughout this sequence, $\blnk$ labels to the right are consumed, while all qudits to the left are labeled $\dead$. This procedure continues until the entire circuit has been simulated.

Lastly, we observe that the final timestep of $U'$ is encoded by \cite{HNN13} in the following configuration:
\begin{align}\label{eqn:configFinal}
	\cdots \dead \dead \bdry \dead \dead \bdry \dead \qubit\parity \insi  \qubit\parity \cdots\parity\insi\qubit\parity \insi \gate \bdry
\end{align}

\paragraph{3. Adding $1$-local ``sifters''.}
We now add $1$-local Hamiltonian terms which serve to ``sift'' through bad proofs, or more accurately to split the ground space of $G$, so as to force low-energy states to encode correct query answers. As previously described, even a correct simulation of the circuit $V$ may not output the correct answer for instance $\Pi$ if a malicious prover supplies incorrect proofs to the query registers $\spa{Y}_i$; in particular, a prover might send a proof which accepts with low probability even though $\pi_i$ is a YES-instance.
Intuitively, we wish to penalize states encoding a proof $\ket{\psi_i}$ which leads verifier $V_i$ to reject with high probability when there exists a proof $\ket{\phi_i}$ such that $V_i$ would have accepted with high probability (here, query $\pi_i$ is a YES instance). For answer register $\spa{X}_i$, we add a ``sifter'' penalty term $\epsilon\ketbra{0}{0}_{\spa{X}_i}$, for $\epsilon$ some inverse polynomial to be set later.
These terms are similar to the $\hout$ term from other Hamiltonian constructions; but, here we are not only concerned about the ground space but also about the low-energy space.
As in other constructions, we must penalize NO answers enough to ensure the ground space encodes YES answers when possible. But, given a correct NO answer, the penalty must be small enough that the energy is gapped lower than any state which encodes an incorrect YES, such as those which by encode an invalid computation leading to YES.

However, because the encoding enforced by $G$ shifts the block of logical qubits rightwards along the line as the computation progresses, the location of a particular logical qubit's encoding depends on the current timestep. Thus, in order to properly act on logical qubit $\spa{X}_i$, we must be careful to specify the configuration which the penalty term acts on.

We may assume that once $V'$ finishes simulating all of the circuits $V_i$, it arranges each of the outputs in the first $m$ logical qubits on the line, finishing by the end of some round $r^*-1$, such that the $i$-th logical qubit on the line is the qubit which $V$ stored in $\spa{X}_i$. (The value of $r^*$ can be determined during the construction of $V'$.)
We may also assume that $V'$ then ``pauses'' by applying only identity gates in round $r^*$.
This round is encoded in block $r^*$, and since each block is comprised of $2n$ qudits, the answers to queries 1 to $m$ are thus simultaneously stored in qudits
\begin{equation}\label{eqn:qi}
    q_i := (2n)(r^* -1)+(2i-1).
\end{equation}
The $m$ sifter terms are given by
\[
	\houti[i] = \epsilon\ketbra{\gate_0}{\gate_0}_{q_i} ,
\]
where the subscript denotes the qudit which the term acts on and $\epsilon$ is to be set later.
Note that there is a unique legal configuration in which any given qudit is labeled $\gate$, so $\houti$ will apply to at most one state $\ket{\psi_t}$ in the history state of Equation~\eqref{eqn:hist-state-1D}.
Finally, we define $\hout = \sum_{i=1}^{m} \houti$.

\paragraph{The final Hamiltonian.} Our final Hamiltonian is
$H := G + \hout=\din\hin+\dprop\hprop+\dpen\hpen+\hout$,
with $\din,\dprop,\dpen$ polynomials to be set later.

\paragraph{The observable.}

Recall the configuration from~\eqref{eqn:configFinal}, which corresponds to the final timestep in the computation of a circuit passed to the construction of \cite{HNN13}. Note that this is the unique timestep in which the final qudit is labeled $\gate$.
We assume, without loss of generality, that $V'$ places its final output in the rightmost logical qubit on the line.
Thus, we choose single-qudit observable $A = \ketbra{\gate_0}{\gate_0}_{2nR}$,
where the subscript denotes that $A$ acts on the rightmost qudit on the line, where $R$ is the number of rounds in $V'$.

\paragraph{Setting parameters.}

Let $L$ denote the number of legal configurations which the history state in \eqref{eqn:hist-state-1D} is summed over, which is at most polynomial in $n$. We have that $H$ is $k$-local and $A$ is $\ell$-local for $k:=2$ and $\ell:=1$.
Set $\epsilon = 1/(8m)$, where recall $m$ is the (polynomial) number of queries.
Then, set $p$, the completeness/soundness error of each $V_i$, to some inverse-exponential in $n$ such that $p<\epsilon$ for all $n$.
Set $a= 1/(4L)$ and $b= 3/(4L)$.
We will set $\delta$ to a sufficiently small fixed inverse polynomial in $n$ in the proof of Lemma~\ref{lem:L}, which will then set $\din,\dprop,\dpen$ to sufficiently large fixed polynomials in $n$ via the proof of Lemma~\ref{lem:props}.

This concludes our deterministic polynomial-time mapping of the input $\Ppar{QMA}$ computation $\Pi$ to the 1D instance $\tilde{\Pi}:=(H,A,k,\ell,a,b,\delta)$ of $\apptwo$.

\subsection{Correctness}\label{ssec:1D-correctness}

We now prove Theorem~\ref{thm:apxsim-1dline-complete} by showing correctness of our construction from Section~\ref{ssec:1d-construction}. A number of lemmas required in the proof are deferred to Section~\ref{sscn:lemmas} to ease the exposition; in particular, we require Lemma~\ref{lem:props}, which explicitly proves two facts about the low-energy space of the construction of \cite{HNN13}, Lemma~\ref{lem:L}, which shows that a history state in our construction must simultaneously encode nearly correct answers for all valid queries $\pi_i$, and Lemma~\ref{lem:unionbound}, which states a Commutative Quantum Union Bound.

\begin{proof}[Proof of Theorem~\ref{thm:apxsim-1dline-complete}]
	Containment in $\PQMA$ was already shown for up to $O(\log n)$-local $H$ by \cite{A14}, with no restriction on the geometry.
	Our goal is now to show $\PQMApar$-hardness, which by Theorem~\ref{thm:APXcompleteness} yields $\PQMA$-hardness.
	We show hardness for the problem $\apptwo$, which recall from Section~\ref{sscn:results2} trivially reduces to $\app$, thus yielding hardness for $\app$.
	Let $\Pi$ be a $\Ppar{QMA}$ computation and map it to the $\apptwo$ instance $\tilde{\Pi}=(H,A,k,l,a,b,\delta)$ as described in Section~\ref{ssec:1d-construction}. The proof proceeds in two parts: We first show that low energy states must necessarily encode correct query answers, and subsequently apply this to show correctness in YES and NO cases for $\Pi$.
	
	\paragraph{Low energy states approximately encode correct query answers.}
	Recall that $H=G+\hout$. Let $\delta,\gamma$ denote arbitrary inverse polynomials in $n$ which will be set later in Lemma~\ref{lem:L}. Consider any state $\ket{\psi}$ such that $\bra{\psi}H\ket{\psi} \leq \lambda(H)+\delta$. Since $\hout\succeq 0$, $\bra{\psi}G\ket{\psi} \leq \lambda(H)+\delta$ as well. By Lemma~\ref{lem:props}, for sufficiently large fixed polynomials $\din,\dprop,\dpen$, two statements thus hold: First, the nullspace $\spa{S}$ of Hamiltonian $G = \din\hin+\dprop\hprop+\dpen\hpen$ is the span of all correctly encoded history states, as defined in Equation~\eqref{eqn:hist-state-1D}; Second, there exists a correctly encoded history state $\ket{\psihist}$ such that
	\begin{equation}\label{eqn:1D-trace-dist}
	\trnorm{\ketbra{\psi}{\psi}-\ketbra{\psihist}{\psihist}}\leq \gamma.
	\end{equation}
	Combining Equation~(\ref{eqn:1D-trace-dist}) with the H\"{o}lder Inequality and the fact that $\snorm{\hout}=m\epsilon$ yields that
	\[
	\abs{ \tr\left[\hout \ketbra{\psi}{\psi}\right] - \tr\left[\hout\ketbra{\psihist}{\psihist}\right] }  \leq \gamma \snorm{\hout} = m\epsilon \gamma .
	\]
	Since $\ket{\psihist}$ is a nullstate of $G$ and $\bra{\psi}\hout\ket{\psi}\leq \bra{\psi}H\ket{\psi}\leq \lambda(H)+\delta$, we conclude
    \begin{equation}\label{eqn:2}
        \bra{\psihist} H \ket{\psihist} \leq \lambda(H) + \delta + m\epsilon\gamma.
    \end{equation}
	
Next, let $I\subseteq\set{1,\ldots,m}$ be the set of indices corresponding to valid queries $\pi_i$, and for all $i\in I$ define $x_i=1$ if $\pi_i$ is a YES-instance and $x_i=0$ if $\pi_i$ is a NO-instance.\footnote{Without loss of generality, we may assume at least one query is valid ($I\neq\emptyset$). This is because if all queries are invalid, then all simulations of the $\Poly$ circuit $C$ must output the same answer no matter the sequence of query answers $C$ receives. Thus, all history states will encode the same final answer, and $\alpha$ (defined after \eqref{eqn:1D-query-string-error}) equals 1, satisfying the lower bound found of $\alpha\geq 1-m\epsilon$.}
Recall now from Section~\ref{ssec:1d-construction} that at the beginning of round $r^*$, $V'$ has encoded the answer to the $i$-th QMA query in qudit $q_i$ (defined in Equation~(\ref{eqn:qi})). Let $\ket{\psi_{t^*}}$ denote the unique (normalized) state in the superposition comprising $\ket{\psihist}$ in which $q_1$ is labeled $\gate$ (i.e. the first timestep corresponding to round $r^*$). Since during round $r^*$, $V'$ only applies identity gates,
the qubits encoded in qudits $q_i$ during timestep $t^*$, in which $q_1$ is labeled $\gate$ and all other $q_i$ are labeled $\qubit$, are exactly the same as in successive timesteps in which other $q_i$ are labeled by $\gate$.
More formally, $ |\brakett{\psi_{t^*}}{\qubit_{x_i}}_{q_i}|^2  = L |\brakett{\psihist}{\gate_{x_i}}_{q_i}|^2$ for any $i\in I$, and so by Lemma~\ref{lem:L},
\begin{equation}\label{eqn:4}
    \abs{ \brakett{\psi_{t^*}}{\qubit_{x_i}}_{q_i} }^2 \geq 1-\epsilon,
\end{equation}
where\footnote{We implicitly apply identity on all qudits other than $q_i$, i.e. $\abs{ \brakett{\psihist}{\qubit_{x_i}}_{q_i} }^2 := \tr\left[ \ketbra{\psihist}{\psihist} \left(I\otimes \ketbra{\qubit_{x_i}}{\qubit_{x_i}}_{q_i} \otimes I\right) \right]$.} we substitute the label $\gate$ for $\qubit$ when $i=1$, and where the factor of $L^{-1}$ is removed due to the normalization of $\ket{\psi_{t^*}}$.

This is for any single query $\pi_i, i\in I$; from this, we can obtain that $\ket{\psi_{t^*}}$ simultaneously encodes nearly correct query answers to \emph{all} valid queries. To do so, define $\Gamma:=\Pi_{i\in I}\ketbra{\qubit_{x_i}}{\qubit_{x_i}}_{q_i}$ (where again, we replace label $\gate$ for $\qubit$ when $i=1$). Then, by the Commutative Quantum Union Bound (Lemma~\ref{lem:unionbound}),
	\begin{equation}\label{eqn:1D-query-string-error}
		\bra{\psi_{t^*}}\Gamma\ket{\psi_{t^*}} \geq 1-\abs{I}\epsilon\geq 1-m\epsilon.
	\end{equation}
It follows that we may write $\ket{\psi_{t^*}} = \alpha \ket{\phi_1} + \beta\ket{\phi_2}$ for unit vectors $\ket{\phi_1}, \ket{\phi_2}$ such that $\Gamma\ket{\phi_1}=\ket{\phi_1}$ and $\Gamma\ket{\phi_2}=0$, and where $\alpha,\beta\in \complex,\abs{\alpha}^2 + \abs{\beta}^2 = 1$, and $\abs{\alpha}^2 \geq 1-m\epsilon$. Intuitively, $\ket{\phi_1}$ is the part of $\ket{\psi_{t^*}}$ that encodes correct strings of query answers on $I$, while $\ket{\phi_2}$ encodes strings with at least one incorrect query answer in $I$ --- for clarity, $\ket{\phi_1}$ may encode a superposition of multiple \emph{distinct} correct strings of query answers, since queries with indices not in $I$ may be answered arbitrarily.
	
	\paragraph{Application to YES versus NO cases for $\Pi$.}
	We have shown that for any low energy state $\ket{\psi}$, there exists a history state $\ket{\psihist}$ close to $\ket{\psi}$ which has large amplitude on all the correct query answers for set $I$ in round $r^*$. We can now analyze the YES and NO cases for our $\PQMA$ problem $\Pi$.

Recall that $\ket{\phi_1}$ may be a superposition over \emph{multiple} correct query strings (due to invalid queries $\pi_i$ for $i\not\in I$). Nevertheless, since the classical circuit $C$ for the $\PQMA$ machine is required to output the \emph{same} answer regardless of how invalid queries are answered (i.e. for any given correct string of query answers), all query strings which $\ket{\phi_1}$ is a superposition over lead $C$ to output the same, correct final answer.
	Thus, setting $y=0$ if $\Pi$ is a YES-instance and $y=1$ if $\Pi$ is a NO-instance, we have
\[
    \abs{\bra{\psihist}A\ket{\psihist} - \frac{y}{L}} \leq \frac{m\epsilon}{L},
\]
where the factor of $L^{-1}$ is due to the fact $A$ applies only to the final configuration/time step. Combining Equation~\eqref{eqn:1D-trace-dist} with the H\"{o}lder inequality yields that
	$\abs{ \tr\left[A \ketbra{\psi}{\psi}\right] - \tr\left[A\ketbra{\psihist}{\psihist}\right] }  \leq \gamma$, since $\snorm{A}=1$,
	and so
	\[
		\abs{\bra{\psi}A\ket{\psi} - \frac{y}{L}} \leq \frac{m\epsilon}{L} +\gamma ,
	\]	
	Given that we set $\delta=\gamma=1/(256m^2L)<1/(8L)$ in Lemma~\ref{lem:L} and $\epsilon = 1/(8m)$, we have that $\gamma + m\epsilon/L < 1/(4L)$.
	We conclude that for all low-energy states $\ket{\psi}$ (i.e. states satisfying $\bra{\psi}H\ket{\psi}\leq \lambda(H)+\delta$), if $\Pi$ is a YES-instance then $\bra{\psi}A\ket{\psi} \leq 1/(4L)$ (i.e. we have a YES instance of $\apptwo$), and if $\Pi$ is a NO-instance then $\bra{\psi}A\ket{\psi} \geq 3/(4L)$ (i.e. we have a NO instance of $\apptwo$), as desired.
\end{proof}

\subsubsection{Required lemmas for proof of Theorem~\ref{thm:apxsim-1dline-complete}}\label{sscn:lemmas}
We begin by restating a known lemma and corollary.

\begin{lemma}[Kempe, Kitaev, Regev~\cite{KKR06}]\label{l:proj}
	Let $H=H_1+H_2$ be the sum of two Hamiltonians operating on some Hilbert space $\spa{H}=\spa{S}+\spa{S}^\perp$. The Hamiltonian $H_1$ is such that $\spa{S}$ is a zero eigenspace and the eigenvectors in $\spa{S}^\perp$ have eigenvalue at least $J>2\snorm{H_2}$. Then,
	\[
		\lambda(H_2|_{\spa{S}})-\frac{\snorm{H_2}^2}{J-2\snorm{H_2}}\leq \lambda(H)\leq \lambda(H_2|_{\spa{S}}),
	\]
where recall $\lambda(H_2|_{\spa{S}})$ denotes the smallest eigenvalue of $H_2$ restricted to space $\spa{S}$.
\end{lemma}
\begin{corollary}[\cite{GY16}]\label{cor:kkr}
	Let $H=H_1+H_2$ be the sum of two Hamiltonians operating on some Hilbert space $\spa{H}=\spa{S}+\spa{S}^\perp$. The Hamiltonian $H_1$ is such that $\spa{S}$ is a zero eigenspace and the eigenvectors in $\spa{S}^\perp$ have eigenvalue at least $J>2\snorm{H_2}$. Let $K:=\snorm{H_2}$. Then, for any $\delta\geq0$ and vector $\ket{\psi}$ satisfying $\bra{\psi}H\ket{\psi}\leq \lambda(H)+\delta$, there exists a $\ket{\psi'}\in \spa{S}$ such that
	\[
\trnorm{\ketbra{\psi}{\psi}-\ketbra{\psi'}{\psi'}}\leq 2
	\left(\frac{K+\sqrt{K^2+\delta(J-2K)}}{J-2K}\right).
\]	
\end{corollary}

\noindent We now prove the lemmas required for Theorem~\ref{thm:apxsim-1dline-complete}.

\begin{lemma}\label{lem:props}
    Assume the notation of Section~\ref{ssec:1d-construction}. For $G = \din\hin+\dprop\hprop+\dpen\hpen$, the following hold:
    \begin{enumerate}
        \item For sufficiently large (efficiently computable) polynomials $\din,\dprop,\dpen$, the null space of $G$ is the span of all correctly encoded history states, i.e. of the form in Equation~(\ref{eqn:hist-state-1D}).
        \item For any fixed inverse polynomials $\delta$ and $\gamma$, there exist efficiently computable polynomials $\din,\dprop,\dpen$ such that for any $\ket{\psi}$ attaining $\bra{\psi}G\ket{\psi}\leq\lambda(G)+\delta$, there exists a correctly encoded history state $\ket{\psihist}$ such that
    \[
        \trnorm{\ketbra{\psi}{\psi}-\ketbra{\psihist}{\psihist}}\leq \gamma.
    \]
    \end{enumerate}
\end{lemma}
\begin{proof}
    The analysis of $G$ is more subtle than that of, say, the $5$-local Kitaev circuit-to-Hamiltonian construction~\cite{KSV02}. The latter required the analysis of two orthogonal subspaces acted on invariantly by the Hamiltonian in question; the span of all correctly encoded history states, and the span of all states with an incorrectly encoded clock register (i.e. illegal configurations). In~\cite{HNN13}, however, due to the restrictions of encoding in 1D, there are \emph{two} types of illegal configurations which can arise --- those which are detectable by local checks, and those which are not --- and $G$ does not act invariantly on the spaces of legal and illegal configurations. The soundness analysis of the $\QMA$-hardness construction of~\cite{HNN13} (see Section 6 therein, which we follow below) hence independently analyzes \emph{three} types of subspaces which are acted on invariantly by $\hprop$: (1) The span of legal configurations and certain locally detectable illegal configurations, (2) the span of certain other locally detectable illegal configurations, and (3) the span of illegal configurations which are not locally detectable. We shall henceforth refer to these subspaces as $S_1$, $S_2$, and $S_3$, respectively.\\

     \noindent\emph{Proof of claim 1.} This claim is implicit in~\cite{HNN13}; we sketch a proof to make it explicit here.
     Claim 2 of~\cite{HNN13} and the subsequent discussion explicitly show that any valid history state is a null state of $G$.
     For the reverse containment,
     Section 6.2 of~\cite{HNN13} shows that for sufficiently large polynomials $\din,\dprop,\dpen$, $\lambda((\dprop\hprop+\dpen\hpen)\vert_{S_3})\in\Omega(1)$. That $\lambda(G\vert_{S_2})\geq \dpen$ follows since $\hpen$ is a sum of pairwise commuting projectors. Thus, $\Null(G)$ resides in $S_1$. Section 6.1 of~\cite{HNN13} shows that $\Null(\hprop\vert_{S_1\cap\Null(\hpen)})$ is spanned by valid history states. We conclude that the span of all valid history states contains $\Null(G)$.\\

     \noindent\emph{Proof of claim 2.} We know from claim 1 that $\Null(G)$ is precisely the span of all correctly encoded history states. Let $\spa{C}$ denote the orthogonal complement of $\Null(G)$. Then, we know from the proof of claim 1 that $\lambda(G\vert_{\spa{C}\cap S_2})\geq\dpen\in\Omega(1)$, and that $\lambda((\dprop\hprop+\dpen\hpen)\vert_{\spa{C}\cap S_3})\in\Omega(1)$. (Here we have used the fact that $S_2\cup S_3\subseteq \spa{C}$.) Since $\delta$ is assumed to be inverse polynomial in $n$, and since we know from claim $1$ that $\lambda(H)\leq 0$, it follows that no vector $\ket{\psi}$ from $S_2$ or $S_3$ can attain $\bra{\psi}G\ket{\psi}\leq\lambda(G)+\delta$.

     We are thus reduced to the case $\ket{\psi}\in S_1$, which we prove using three applications of Corollary~\ref{cor:kkr}. (To reduce notation, in the remainder of this proof all operators are implicitly restricted to $S_1$.) In the first application, let $H_1=\dpen\hpen$ and $H_2=\din\hin+\dprop\hprop$. Suppose $\bra{\psi}H_1+H_2\ket{\psi}\leq\lambda(H)+\delta$. Then by Lemma~\ref{cor:kkr}, there exists a vector $\ket{\psi'}\in{\Null(\hpen)}$ such that
	\[
\trnorm{\ketbra{\psi}{\psi}-\ketbra{\psi'}{\psi'}}\leq 2
	\left(\frac{K_1+\sqrt{K_1^2+\delta(J_1-2K_1)}}{J_1-2K_1}\right)=:2
	\gamma_1,
\]
    for $K_1:=\snorm{H_2}$ and $J_1>2 K_1$. (Note that since $\dpen\hpen$ is a sum of commuting projectors, its smallest non-zero eigenvalue is at least $\dpen$, i.e. $J\geq \dpen$.) By the H\"{o}lder inequality,
    \begin{equation}\label{eqn:p1}
        \abs{\trace((H_1+H_2)\ketbra{\psi}{\psi})-\trace((H_1+H_2)\ketbra{\psi'}{\psi'})}\leq2\gamma_1\snorm{H_1+H_2}=:\epsilon_1.
    \end{equation}
    Combining these facts, we have
    \begin{eqnarray}
        \bra{\psi'}(H_1+H_2)\vert_{\Null(\hpen)}\ket{\psi'}&=&\bra{\psi'}(H_1+H_2)\ket{\psi'}\nonumber\\
        &\leq&\lambda((H_1+H_2))+\delta+\epsilon_1\nonumber\\
        &\leq&\lambda((H_1+H_2)\vert_{\Null(\hpen)})+\delta+\epsilon_1\nonumber\\
        &=:&\lambda((H_1+H_2)\vert_{\Null(\hpen)})+\delta_2,\label{eqn:p2}
    \end{eqnarray}
    where the first statement holds since $\ket{\psi'}\in \Null(\hpen)$, the second by Equation~(\ref{eqn:p1}), and the third by the Projection Lemma (this follows directly since projections can only increase the smallest eigenvalue).

    We now repeat the process for $H_1=\dprop\hprop\vert_{\Null(\hpen)}$ and $H_2=\din\hin\vert_{\Null(\hpen)}$. The key observation (used also in~\cite{HNN13}) is that restricted to $S_1\cap\Null(\hpen)$, $\hprop$ is now positive semidefinite, has a 1-dimensional null space spanned by the correct history state (the action of $\hprop$ ignores the initial setting of ancilla qubits, including the proof register, which in general leads to multiple correct history states), and its smallest non-zero eigenvalue is at least $1/(2(L+1)^2)$ (recall $L$ is the number of time steps a valid history state sums over). Thus, by Lemma~\ref{cor:kkr}, there exists a vector $\ket{\psi''}\in{\Null(\hpen)\cap\Null(\hprop)}$ such that
	\[
\trnorm{\ketbra{\psi'}{\psi'}-\ketbra{\psi''}{\psi''}}\leq 2
	\left(\frac{K_2+\sqrt{K_2^2+\delta_2(J_2-2K_2)}}{J_2-2K_2}\right)=:2
	\gamma_2,
\]
    for $K_2:=\snorm{H_2}$ and $J_2>2 K_2$. Note that $J_2\geq \dprop/(2(L+1)^2)$. By the H\"{o}lder inequality,
    \[
        \abs{\trace((H_1+H_2)\ketbra{\psi'}{\psi'})-\trace((H_1+H_2)\ketbra{\psi''}{\psi''})}\leq2\gamma_2\snorm{H_1+H_2}=:\epsilon_2,
    \]
    which yields
    \begin{eqnarray*}
        \bra{\psi''}(H_1+H_2)\vert_{\Null(\hprop)}\ket{\psi''}&=&\bra{\psi''}(H_1+H_2)\ket{\psi''}\nonumber\\
        &\leq&\lambda((H_1+H_2))+\delta_2+\epsilon_2\nonumber\\
        &\leq&\lambda((H_1+H_2)\vert_{\Null(\hprop)})+\delta_2+\epsilon_2\nonumber\\
        &=:&\lambda((H_1+H_2)\vert_{\Null(\hprop)})+\delta_3.
    \end{eqnarray*}
     Finally, we repeat the process for $H_1=\din\hin\vert_{\Null(\hpen)\cap\Null(\hprop)}$ and $H_2=0$. Since by claim 1 we know the joint null space of $\hin,\hprop,\hpen$ is non-empty, by Lemma~\ref{cor:kkr}, there exists a vector $\ket{\psi'''}\in{\Null(\hpen)\cap\Null(\hprop)\cap\Null(\hin)}$ such that
	\[
\trnorm{\ketbra{\psi''}{\psi''}-\ketbra{\psi'''}{\psi'''}}\leq 2
	\sqrt{\frac{\delta_3 }{J_3}}=:2
	\gamma_3,
\]
    for $J_3>0$. Note that $J_3\geq \din$ since $\hin$ is a sum of commuting projectors. By claim 1, since $\ket{\psi'''}$ is in the joint null space of $\hin,\hprop,\hpen$, it is a correctly encoded history state; denote it $\ket{\psihist}$. By the triangle inequality we have
    \[
        \trnorm{\ketbra{\psi}{\psi}-\ketbra{\psihist}{\psihist}}\leq 2(\gamma_1+\gamma_2+\gamma_3).
    \]
    The claim now follows by observing that all variables involved, i.e. $\delta_2,\delta_3,\epsilon_1,\epsilon_2,\gamma_1,\gamma_2,\gamma_3,J_1,J_2,J_3$, decrease inverse polynomially in (a non-empty subset of) polynomials $\din,\dprop,\dpen$. Thus, for any desired target accuracy $q$, we may attain the claim by setting $\din,\dprop,\dpen$ as sufficiently large polynomials. (Note that this requires upper bounding terms of the form $K_2:=\snorm{H_2}$, which is easily done via triangle inequality of the spectral norm and the fact that projections can only decrease maximum eigenvalues.)
\end{proof}

\begin{lemma}\label{lem:L}
	Assume the notation of Section~\ref{ssec:1D-correctness}. For all $i\in I$, it holds that
	\begin{equation}\label{eqn:1D-answer-bound}
	\abs{ \brakett{\psihist}{\gate_{x_i}}_{q_i} }^2 \geq \frac{1-\epsilon}{L} ,
	\end{equation}
	where recall $q_i$ is the index of the qudit which encodes the output corresponding to query $\pi_i$ following the verification phase.
\end{lemma}

\begin{proof}
	For clarity, the factor of $L^{-1}$ comes from the $L$ configurations which $\ket{\psihist}$ is a sum over.
	Recall there is a unique configuration in which any given qudit is labeled $\gate$, implying all history states $\ket{\psihist}$ satisfy
	\begin{equation}\label{eqn:1D-gate-sum}
		\abs{\brakett{\psihist}{\gate_0}_{q_i}}^2 + \abs{\brakett{\psihist}{\gate_1}_{q_i}}^2 = \frac{1}{L} .	
	\end{equation}
	We prove our claim by contradiction via an exchange argument. Suppose there exists a valid query\footnote{If all queries are invalid, then Lemma~\ref{lem:L} holds vacuously.}
	$\pi_j$ with correct answer $x_j$ such that
\[
\abs{ \brakett{\psihist}{\gate_{x_j}}_{q_j} }^2 < \frac{1-\epsilon}{L}.
\]
Since $\ket{\psihist}$ is a correctly encoded history state, we claim $\pi_j$ must be a YES-instance. For if $\pi_j$ were a NO-instance, then all simulations of $V_j$ (on any possible proof) output NO with probability at least $1-p$. Thus, $\ket{\psihist}$ always encodes an output qubit such that
\[
    \abs{\brakett{\psihist}{\gate_{0}}_{q_j}}^2 \geq \frac{1-p}{L} \geq \frac{1-\epsilon}{L},
\]
which would contradict our supposition.
	
Given that $\pi_j$ is a YES-instance, we have that $\abs{\brakett{\psihist}{\gate_{1}}_{q_j}}^2 \leq (1-\epsilon)/L$, and so by Equation~(\ref{eqn:1D-gate-sum}), $\bra{\psihist}\houti[j]\ket{\psihist} \geq \epsilon^2 /L$.
	Further, since $\pi_j$ is a YES-instance, there exists a QMA proof $\ket{\omega}$ which causes $V_j$ to output YES with probability at least $1-p$.
	By exchanging the QMA proof which $\ket{\psihist}$ encodes for circuit $V_j$ with the proof $\ket{\omega}$, we obtain a new history state $\ket{\psihist'}$ which satisfies
\[
    \abs{ \brakett{\psihist'}{\gate_{1}}_{q_j} }^2 \geq \frac{1-p}{L},
\]
and so $\bra{\psihist'}\houti[j]\ket{\psihist'} \leq p\epsilon/L$. Hence,
\begin{equation}\label{eqn:houtj-bound}
   \bra{\psihist}\houti[j]\ket{\psihist}-\bra{\psihist'}\houti[j]\ket{\psihist'}\geq \frac{(\epsilon-p)\epsilon}{L},
\end{equation}
i.e. flipping the incorrect query answer saves a non-trivial energy penalty on $\houti[j]$.

We now use this to obtain the desired contradiction. Recall that $H = G + \hout$. We make two observations:
First, because all the $\QMA$ queries are made in parallel, flipping the answer to query $\pi_j$ does not affect the other queries the $\Poly$ machine makes or the answers it receives. Thus, $\ket{\psihist}$ and $\ket{\psihist'}$ obtain the same energy on all terms of $\hout$ other than $\houti[j]$, and Equation~\eqref{eqn:houtj-bound} holds for $\hout$ in place of $\houti[j]$. (Analyzing adaptive queries, rather than parallel, would require that penalties for later queries be carefully weighted less than penalties for earlier queries~\cite{A14}, leading to a significantly more involved analysis.)
Second, both $\ket{\psihist}$ and $\ket{\psihist'}$ are null states of $G$, and so we may substitute $H$ for $\hout$, yielding
\begin{equation}\label{eqn:3}
   \bra{\psihist}H\ket{\psihist}-\bra{\psihist'}H\ket{\psihist'}\geq \frac{(\epsilon-p)\epsilon}{L}.
\end{equation}
Now, recall from Equation~(\ref{eqn:2}) that $\bra{\psihist} H \ket{\psihist} \leq \lambda(H) + \delta + m\epsilon\gamma$.
Since $\delta$ and $\gamma$ are inverse polynomials which (by Lemma~\ref{lem:props}) we are free to choose as needed (the choice of $\delta$ and $\gamma$, in turn, will mandate the choices of $\din,\dprop,\dpen$ via Lemma~\ref{lem:props}), we set $\delta=\gamma=1/(256m^2L)$ (where recall $L$ and $m$ are fixed polynomials in $n$).
These choices of $\delta,\gamma$ satisfy $\delta+m\epsilon\gamma<(\epsilon-p)\epsilon/L$, which combined with Equation~(\ref{eqn:3}) gives that $\bra{\psihist} H\ket{\psihist} > \lambda(H) + \delta + m\epsilon\gamma$, i.e. $\ket{\psihist}$ could not have been close to the ground state energy of $H$. Hence, we have a contradiction, completing the proof.
\end{proof}

Finally, we require a known quantum analogue of the union bound for commuting operators (see, e.g.~\cite{OMW19}). Generalizations to \emph{non-commuting} projectors are given in~\cite{S12,G15,OMW19}.
\begin{lemma}[Commutative Quantum Union Bound]\label{lem:unionbound}
    Let $\set{P_i}_{i=1}^m$ be a set of pairwise commuting projectors, each satisfying $0\preceq P_i\preceq I$. Then for any quantum state $\rho$,
    \[
        1-\trace(\Pi_m\cdots P_1\rho P_1\cdots \Pi_m)\leq\sum_{i=1}^m\trace((I-P_i)\rho).
    \]
\end{lemma}
\noindent The simple proof of Lemma~\ref{lem:unionbound} is given in Appendix~\ref{scn:union} for completeness.

%%%%%%%%%%%%%%%%%%%%%%%%%%%%%%%%%%%%%%
%%%%%%%%%%%%%%%%%%%%%%%%%%%%%%%%%%%%%%

\subsection*{Acknowledgments}
We are grateful to Thomas Vidick for helpful discussions which helped initiate this work. We also thank an anonymous referee for~\cite{GY16} (written by two of the present authors) for the suggestion to think about 1D systems. SG acknowledges support from NSF grants CCF-1526189 and CCF-1617710. SP was supported by EPSRC. Part of this work was completed while JY was supported by a Virginia Commonwealth University Presidential Scholarship. JY acknowledges QIP 2019 student travel funding (NSF CCF-1840547).

%=========================================================================================
%  BIBLIOGRAPHY
%=========================================================================================
\bibliographystyle{alpha}
\bibliography{GPY19_arxiv_submitted}

\appendix

\section{General simulations}
\label{sec:generalsimulations}

In this section we will give a full proof of Lemma~\ref{lem:simulations=reductions} and show that \emph{any} efficient simulation will preserve hardness of $\apptwo$, not just the special case considered in Definition~\ref{dfn:specialsim}.
To state the full definition of simulation, we must first introduce the notion of an encoding.

\begin{definition}[\cite{CMP18}]
\label{dfn:localencoding}
We say a map $\cE:\mathcal{B}(\mathcal{H})\rightarrow \mathcal{B}(\mathcal{H}')$ is an encoding if it is of the form
\[
    \cE(M)=V(M\otimes P +\overline{M} \otimes Q) V^{\dagger}
\]
where $\overline{M}$ denotes the complex conjugate of $M$, $P$ and $Q$ are orthogonal projectors (i.e. $PQ=0$) on an ancilla space $E$; and $V$ is an isometry $V: \mathcal{H}\otimes E \rightarrow \mathcal{H}'$.

When $\cH$ is a many body system with a decomposition $\cH= \bigotimes_{i=1}^n \cH_i$, we say $\cE$ is a local encoding if $E=\bigotimes_{i=1}^n E_i$ such that:
\begin{itemize}
\item $V=\bigotimes_{i=1}^n V_i$ where each $V_i$ acts on $\mathcal{H}_i \otimes E_i$.
\item for each $i$, there exist orthogonal projectors $P_{E_i}$ and $Q_{E_i}$ on $E$ which act non-trivially only on $E_i$, and satisfy $PP_{E_i}=P$ and $QQ_{E_i}=Q$.
\end{itemize}
\end{definition}

\noindent We are now ready to give the full definition of simulation.
\begin{definition}[\cite{CMP18}]
  \label{dfn:sim}
  We say that $H'$ is a $(\Delta,\eta,\epsilon)$-simulation of $H$ if there exists a local encoding $\cE(M)=V(M\otimes P + \overline{M}\otimes Q)V^{\dagger}$ such that:
  \begin{enumerate}
  \item
    There exists an isometry $\widetilde{V}: \mathcal{H}\otimes E \rightarrow \mathcal{H}'$ such that  $\|\widetilde{V} - V\| \le \eta$; and that the encoding $\widetilde{\cE}(M)=\widetilde{V}(M\otimes P + \overline{M}\otimes Q)\widetilde{V}^{\dagger}$ satisfies  $\widetilde{\cE}(I)=P_{\le \Delta(H')}$.
  \item
    $\| H'_{\le \Delta} - \widetilde{\mathcal{E}}(H)\| \le \epsilon$.
  \end{enumerate}
  We say that a family $\mathcal{F}'$ of Hamiltonians can simulate a family $\mathcal{F}$ of Hamiltonians if, for any $H \in \mathcal{F}$ and any $\eta,\epsilon >0$ and $\Delta \ge \Delta_0$ (for some $\Delta_0 > 0$), there exists $H' \in \mathcal{F}'$ such that $H'$ is a $(\Delta,\eta,\epsilon)$-simulation of $H$.
  We say that the simulation is efficient if, in addition, for $H$ acting on $n$ qudits, $\|H'\| = \poly(n,1/\eta,1/\epsilon,\Delta)$; $H'$ and $\set{V_i}$ are efficiently computable given $H$, $\Delta$, $\eta$ and $\epsilon$; and each local isometry $V_i$ in the decomposition $V=\bigotimes_i V_i$ maps to $O(1)$ qudits.
\end{definition}

\noindent We note that Definition~\ref{dfn:specialsim} is just the special case of Definition~\ref{dfn:sim} where $\cE(M)=VMV^{\dagger}$.
We are now ready to restate and prove Lemma~\ref{lem:simulations=reductions}.

\begin{replemma}{lem:simulations=reductions}[Simulations preserve hardness of $\apptwo$]
Let $\mathcal{F}$ be a family of Hamiltonians which can be efficiently simulated by another family $\mathcal{F}'$.
Then $\Fapptwo$ reduces to $\Fapptwo[F']$.
\end{replemma}

\begin{proof}
For brevity, let $P_{\le \Delta}:=P_{\le \Delta(H')}$. Let $\rho'=\proj{\psi'}$ be a state on $\cH'$ such that $\bra{\psi'}H'\ket{\psi'} \le \delta'$ and let $\widetilde{\rho}=P_{\le \Delta} \rho'P_{\le \Delta}/ \tr(P_{\le \Delta}\rho')$, so that by Lemma~\ref{lem:lowenergydistance}, we have $\| \rho' - \widetilde{\rho}\|_1 \le 2 \sqrt{\frac{\delta'}{ \Delta-\lambda(H')}}$.

Since $P_{\le \Delta}$ commutes with $H'$, we have
\begin{align}
\tr(H'\rho')&=\tr(H'P_{\le \Delta} \rho' P_{\le \Delta}) +\tr(H'(I-P_{\le \Delta})\rho' (I- P_{\le \Delta}))\noindent \\
&=p \tr(H'\widetilde{\rho}) + (1-p) \tr(H' \widetilde{\rho}^{\perp}) \ge \tr(H'\widetilde{\rho}), \label{eqn:1}
\end{align}
where $p= \tr(P_{\le \Delta}\rho')$, $\widetilde{\rho}^{\perp}= (I-P_{\le \Delta}) \rho'(I-P_{\le \Delta})/ \tr((I-P_{\le \Delta})\rho')$, and the final inequality follows because $\tr(H' \widetilde{\rho}^{\perp}) \ge \Delta \ge \tr(H' \widetilde{\rho})$.

Now let
	\[\rho= \tr_E\left(\Vt^{\dagger}\widetilde{\rho} \Vt(I \otimes P)\right) + \tr_E\left(\overline{\Vt^{\dagger}\widetilde{\rho} \Vt(I \otimes Q)}\right)\]
	and note that for any operator $A$ on $\cH$, we have
\[
\tr(\widetilde{\mathcal{E}}(A)\widetilde{\rho})
=\tr\left(\Vt (A\otimes P +\overline{A} \otimes Q) \Vt^{\dagger}\widetilde{\rho}\right)
=\tr\left(A\otimes P \Vt^{\dagger} \widetilde{\rho} \Vt\right)+\tr\left(\overline{A}\otimes Q \Vt^{\dagger} \widetilde{\rho} \Vt\right)
=\tr(A\rho).
\]
Therefore,
\[
\tr(H\rho)
=\tr(\widetilde{\cE}(H) \widetilde{\rho})
\le \tr(H'\widetilde{\rho}) + \|H'_{\le \Delta}-\widetilde{\cE}(H)\|
\le \tr(H' \rho') + \epsilon
\le \lambda(H') + \delta' + \epsilon
\le \lambda(H)+\delta' + 2\epsilon,
\]
where the second inequality follows from Equation~(\ref{eqn:1}) and the last inequality from Lemma 27 of \cite{CMP18}, which roughly states that eigenvalues are preserved up to additive error $\epsilon$ in a simulation.

At this point the proof diverges from the simpler case because $\rho$ may be a mixed state, even when $\rho'=\proj{\psi'}$ is pure. Despite having a bound on $\tr(H\rho)$, this bound may not hold for all pure states in the spectral decomposition of $\rho$. Let $\rho_{\delta}= P_{\delta} \rho P_{\delta}/ \tr(P_{\delta})$, where $P_{\delta}$ is the projector onto eigenvectors of $H$ with energy less than $\delta$. By Lemma~\ref{lem:lowenergydistance}, $\|\rho- \rho_{\delta}\|_1\le 2\sqrt{\frac{\delta'+2\epsilon}{\delta}}$.
We will use the spectral decomposition of $\rho_{\delta}= \sum_i \mu_i \proj{\phi_i}$ where the $\ket{\phi_i}$ are orthogonal states with energy $\bra{\phi_i}H\ket{\phi_i} \le \lambda(H)+\delta$ and thus, for observable $A$ given as part of of $\Fapptwo$ input,
\[
    \tr(A \rho_{\delta})= \sum_i \mu_i \bra{\phi_i}A \ket{\phi_i} \quad \left\{
\begin{array}{l} \le a \text{ in a YES instance}\\
\ge b \text{ in a NO instance.} \end{array}\right.
\]

Let $U=V\Vt^{\dagger}$ , which satisfies $U \widetilde{\cE}(A)= \cE(A)U$ for any $A$, and so $\cE(I) U \widetilde{\rho} U^{\dagger}=U \widetilde{\cE}(I) \widetilde{\rho}U^{\dagger}= U \widetilde{\rho} U^{\dagger}$.
Now we need to choose $A'$ such that $A'\cE(I)=\cE(A)$. (Two notes: First, $\cE(I)\neq I$ necessarily, as $P$ and $Q$ need not sum to identity. Second, setting $A'=\cE(A)$ is not necessarily desirable, as $P$ and $Q$ may be non-local projectors.) For example if $A=B_i \otimes I$, let $A'=V_i(B_i \otimes P_{E_i} + \overline{B_i} \otimes Q_{E_i})V_i^{\dagger}\otimes I$. We note that the locality of $A'$ depends on the number of qudits which $V_i$ maps to, which is $O(1)$ by the definition of efficient simulation.
Then
\[\tr(A\rho)=\tr\left(\widetilde{\cE}(A) \widetilde{\rho}\right)=\tr\left(\cE(A) U \widetilde{\rho} U^{\dagger} \right)=\tr(A'\cE(I)U \widetilde{\rho} U^{\dagger})=\tr(A'U \widetilde{\rho} U^{\dagger})\]
and therefore
\begin{align*}
|\tr(A' \rho')-\tr(A\rho_{\delta})| & \le |\tr(A' \rho')-\tr(A'U \widetilde{\rho} U^{\dagger})|+|\tr(A \rho)-\tr(A\rho_{\delta})|\\
&\le \|A'\| \left(\|\rho'-\widetilde\rho\|_1 + \|\widetilde{\rho}-U \widetilde{\rho} U^{\dagger} \|_1\right) + \|A\| \| \rho- \rho_{\delta}\|_1\\
&\le \|A\| \left( 2 \sqrt{\frac{\delta'}{ \Delta-\lambda(H')}}+2\eta +2\sqrt{\frac{\delta'+2\epsilon}{\delta}} \right),
\end{align*}

We note that $\|\widetilde{\rho}-U\widetilde{\rho}U^{\dagger}\|_1 \le 2 \eta$ follows from $\|U-\Vt\Vt^{\dagger}\|\le \eta$, and that $\Vt\Vt^{\dagger} \widetilde{\rho} = P_{\le \Delta} \widetilde{\rho}= \widetilde{\rho}$. Therefore we just need to choose $\Delta, \epsilon, \eta, \delta'$ such that this is less than $(b-a)/3$ and then set $a'=a+(b-a)/3 $ and $b'=b-(b-a)/3$.

\end{proof}

\section{Proof of commutative quantum union bound}\label{scn:union}
\begin{replemma}{lem:unionbound}[Commutative Quantum Union Bound]
    Let $\set{P_i}_{i=1}^m$ be a set of pairwise commuting projectors, each satisfying $0\preceq P_i\preceq I$. Then for any quantum state $\rho$,
    \[
        1-\trace(P_m\cdots P_1\rho P_1\cdots P_m)\leq\sum_{i=1}^m\trace((I-P_i)\rho).
    \]
\end{replemma}
\begin{proof}
    We proceed by induction on $m$. The case of $m=1$ is trivial. Consider $m>1$. Since the $P_i$ pairwise commute, $\trace(P_m\cdots P_1\rho P_1\cdots P_m)=\trace(P_m\cdots P_1\rho):=\trace(P_m M\rho)$ for brevity, and $M$ is a projector. Then,
    \begin{eqnarray*}
        1-\trace(P_m M\rho)&=&\trace((I-P_m)M\rho)+\trace(P_m(I-M)\rho)+\trace((I-P_m)(I- M)\rho)\\
        &=&\trace((I-P_m)\rho)+\trace((I-M)\rho)-\trace((I-P_m)(I-M)\rho)\\
        &\leq &\trace((I-P_m)\rho)+\trace((I-M)\rho),
    \end{eqnarray*}
where the second equality holds since $ \trace((I-P_m)(I- M)\rho)$ equals
\[
   \trace((I-P_m)\rho)+\trace((I-M)\rho) - \left(\trace((I-P_m)M\rho)+\trace(P_m(I-M)\rho)+\trace((I-P_m) (I-M)\rho)\right).
\]
Applying the induction hypothesis completes the proof.
\end{proof}
\end{document}